\documentclass[letterpaper, 11pt]{article}
\usepackage{appendix}
 
\newcommand{\comment}[1]{}
 
\usepackage{graphicx}

\usepackage{subfigure}
\usepackage{amsmath}
\usepackage{amssymb}
\usepackage{mdwlist}

\usepackage{hyperref}

\usepackage{xspace}
\usepackage{setspace}

\newcommand{\sa}{{\mathcal A}}

\newcommand{\sv}{V}

\newenvironment{proof}{\noindent {\bf Proof:}~}{\hspace*{\fill}\(\Box\)}
\newenvironment{proofSketch}{\noindent{\bf Proof Sketch:}}{\hspace*{\fill}\(\Box\)}
\newtheorem{theorem}{Theorem}

\newtheorem{claim}{Claim}

\newtheorem{definition}{Definition}

\newtheorem{lemma}{Lemma}
\def\noflash#1{\setbox0=\hbox{#1}\hbox to 1\wd0{\hfill}}
\newcommand{\vectorv}{{\normalfont\textbf{v}}}
\newcommand{\matrixm}{\textbf{M}}

\newcommand{\vectorl}{{\normalfont\textbf{l}}}
\newcommand{\vectorr}{{\normalfont\textbf{r}}}
\newcommand{\bfB}{{\bf B}}

\newcommand{\D}{{\textbf d_H}}
\newcommand{\LL}{{\textbf L}}
\newcommand{\ddd}{{\textbf d_E}}

\newcommand{\HH}{{\mathcal H}}

\newcommand{\msg}{{\tt MSG}}

\newcommand{\bfA}{{\bf A}}

\newcommand{\bfM}{{\bf M}}
\newcommand{\bfP}{{\bf P}}

\newcommand{\fbar}{{\mathcal F}}

\begin{document}
\title{Asynchronous Convex Consensus in the Presence of Crash Faults \footnote{\normalsize This research is supported in part by National Science Foundation awards 1059540 and 1329681. Any opinions, findings, and conclusions or recommendations expressed here are those of the authors and do not necessarily reflect the views of the funding agencies or the U.S. government.}\footnote{\normalsize A version of this work is published in PODC 2014.}}

\author{Lewis Tseng$^{1}$, and Nitin Vaidya$^{2}$\\~\\
 \normalsize $^1$ Department of Computer Science,\\
 \normalsize $^2$ Department of Electrical and Computer Engineering, 
% and\\ \normalsize $^3$ Coordinated Science Laboratory
\\ \normalsize University of Illinois at Urbana-Champaign\\~\\ \normalsize Email: \{ltseng3, 
%nhv\}@illinois.edu \\~\\ Phone: +1 217-244-6024, +1 217-265-5414}
nhv\}@illinois.edu
\\ \normalsize Mailing address: Coordinated Science Lab., 1308 West Main St., Urbana, IL 61801, U.S.A.} % \\~\\ \normalsize Technical Report
%}

\date{February 13, 2014\footnote{\normalsize Modifed August 2015 to add Lemma \ref{lemma:matrix_comp}.}}
\maketitle

%\noindent
%{\bf
%Key Words}: Byzantine failure, Convex consensus, Asynchronous network, Optimal algorithm

\comment{
\noindent
{\bf
\begin{itemize}
\item This is a regular paper submission.
\item If the paper is not selected for a regular presentation, then please consider it as a brief announcement.
\item Lewis Tseng is a full-time Ph.D. student at the University of Illinois\\at Urbana-Champaign. 
\item This paper is eligible for the best student paper award.
\end{itemize}
}
}

~

\begin{abstract}
{\normalsize

This paper defines a new consensus problem, {\em convex consensus}.
Similar to {\em vector consensus} \cite{herlihy_multi-dimension_AA,Vaidya_BVC,Vaidya_incomplete}, the input at each process is 
a $d$-dimensional vector of reals (or, equivalently, a point in the $d$-dimensional Euclidean space).
However, for convex consensus,
the output at each process is 
a {\em convex polytope} contained within the convex hull of the inputs at the
fault-free processes. We explore the {\em convex consensus} problem under crash faults with {\em incorrect} inputs,
and present an {\em asynchronous} approximate convex consensus algorithm
with optimal fault tolerance that reaches consensus on an {\em optimal} output polytope. {\em Convex consensus} can be used to solve other related problems.
For instance, a solution for convex consensus trivially yields a solution for
{\em vector consensus}. More importantly, convex consensus
can potentially be used to solve other more interesting problems, such as {\em convex function optimization} \cite{Boyd_optimization,AA_convergence_markov}.

}
\end{abstract}

\thispagestyle{empty}
\newpage
\setcounter{page}{1}

\section{Introduction}
\label{s_intro}

The distributed {\em consensus} problem has received significant attention over the past three
decades \cite{Welch_textbook}. The traditional consensus problem formulation assumes that each process
has a scalar input.
As a generalization of this problem,
recent work \cite{herlihy_multi-dimension_AA,Vaidya_BVC,Vaidya_incomplete} has addressed 
{\em vector} consensus (also called {\em multidimensional} consensus) in the
presence of Byzantine faults, wherein each process has a $d$-dimensional
vector of reals as input, and the processes reach consensus
on a $d$-dimensional vector within the convex hull of the inputs at fault-free processes ($d\geq 1$).
In the discussion below, it will be more convenient to view a $d$-dimensional vector as
a {\em point} in the $d$-dimensional Euclidean space.

This paper defines the problem of {\em convex consensus}.
Similar to {\em vector consensus}, the input at each process is 
a point in the $d$-dimensional Euclidean space.
However, for convex consensus,
the output at each process is 
a {\em convex polytope} contained within the convex hull of the inputs at the
fault-free processes. Intuitively, the goal is to reach consensus on the ``largest possible''
polytope within the convex hull of the inputs at fault-free processes, allowing the processes
to estimate the domain of inputs at the fault-free processes.
In some cases, the output convex polytope may consist of just a single point, but in general,
it may contain an infinite number of points.

Convex consensus may be used to solve other related problems.
For instance, a solution for convex consensus trivially yields a solution for
{\em vector consensus} \cite{herlihy_multi-dimension_AA,Vaidya_BVC}. More importantly, convex consensus
can potentially be used to solve other more interesting problems, such as {\em convex function optimization} \cite{Boyd_optimization,AA_convergence_markov,Nedic_convex} with the convex hull of the inputs at fault-free
processes as the domain. We will discuss the application of convex consensus to function optimization in Section \ref{s_optimization}.

\paragraph{Fault model:}
With the exception of Section \ref{s_crash_good_inputs}, rest of the paper assumes the
{\em crash faults with incorrect inputs} \cite{Coan_Simulation,Welch_textbook} fault
model. In this model, each faulty process may crash, and may also have an
{\em incorrect input}. A faulty process performs the algorithm faithfully, using possibly
incorrect input, until it (possibly) crashes. The implication of an {\em incorrect input} will be clearer
when we formally define convex consensus below. At most $f$ processes may be faulty.
All fault-free processes have {\em correct inputs}.
% We assume that just knowing the value of an input is not adequate
% to determine whether it is incorrect or correct.
Since this model allows incorrect inputs at faulty processes,
the simulation techniques in \cite{Coan_Simulation,Welch_textbook} can be used to transform an
algorithm designed for this fault model to an algorithm for tolerating Byzantine faults. For brevity, we do not discuss this transformation. (A Byzantine convex consnesus algorithm is also presented
in our technical report \cite{Tseng_BCC}.)
Section \ref{s_crash_good_inputs} briefly discusses how our results extend naturally to
the more commonly used {\em crash fault} model wherein faulty processes have correct
inputs (we will refer to the latter model
as {\em crash faults with correct inputs}).\footnote{Our results also easily extend to the case when up to $f$ processes may
crash, and up to $\psi$ processes may have incorrect inputs, with the
set of crashed processes not necessarily being identical to the processes with
incorrect inputs. For brevity, we omit this generalization.}
% ++++ NEED TO FIND A BETTER PLACE TO CITE OUR TECH REPORT ++++ In our technical report \cite{Tseng_BCC}, we present a Convex Consensus algorithm that is resilient to Byzantine failures using similar simulation technique. 

\paragraph{System model:}
The system under consideration is {\em asynchronous}, and consists of $n$ processes.
 Let the set of processes be denoted as $V=\{1,2,\cdots, n\}$.
 All processes can communicate with each other. Thus, the underlying communication network
 is modeled as a {\em complete graph}.
 Communication channels are reliable and FIFO \cite{AA_Dolev_1986,Coan_Simulation}.
Each message is delivered exactly once on each channel.
The input at process $i$, denoted as $x_i$, is a point
in the $d$-dimensional Euclidean space (equivalently, a $d$-dimensional vector of real numbers).

\paragraph{Convex consensus:}
The FLP impossibility of reaching {\em exact} consensus in asynchronous systems with crash
faults \cite{FLP_one_crash} extends to the problem of convex consensus as well. Therefore,
we consider {\em approximate} convex consensus in our work.
An {\em approximate convex consensus} algorithm must satisfy the
following properties:
\begin{itemize}
\item \textbf{Validity}: The {\em output} (or {\em decision}) at each fault-free process must be
a convex polytope in the convex hull of {\em correct inputs}.
Under the {\em crash fault with incorrect inputs} model, the input at any faulty process may possibly
be {\em incorrect}.
% therefore, the output polytope must be in the convex hull
% of the inputs of the fault-free processes.

\item \textbf{$\epsilon$-Agreement}: For a given constant $\epsilon > 0$, the {\em Hausdorff distance} (defined below) between the output polytopes at any two fault-free processes must be at most $\epsilon$.

\item {\bf Termination:} Each fault-free process must terminate within a finite amount of time.
\end{itemize}

\paragraph{Distance metrics:}
\begin{itemize}
\item $\ddd(p,q)$ denotes the Euclidean distance between points $p$ and $q$. 
All {\em points} and {\em polytopes} in our discussion belong to a $d$-dimensional Euclidean space, for some $d\geq 1$,
even if this is not always stated explicitly.
\item For two convex polytopes $h_1, h_2$, the {\em Hausdorff distance}\, $\D(h_1,h_2)$ is defined as follows \cite{Hausdorff}.
\begin{equation}
\label{e_hausdorf}
\D(h_1, h_2) ~~=~~ \max~~ \{~~ \max_{p_1 \in h_1}~\min_{p_2 \in h_2} \ddd(p_1, p_2),~~~~ \max_{p_2 \in h_2}~\min_{p_1 \in h_1} \ddd(p_1, p_2)~~\}
\end{equation}
\end{itemize}

\paragraph{Optimality of approximate convex consensus:}
The algorithm proposed in this paper is optimal in two ways. It requires an
optimal number of processes to tolerate $f$ faults, and it decides
on a convex polytope that is optimal in a ``worst-case sense'', as elaborated
below:
\begin{itemize}
\item Prior work on approximate {\em vector} consensus mentioned
earlier \cite{herlihy_multi-dimension_AA, Vaidya_BVC} showed that $n\geq (d+2)f+1$
is necessary to solve that problem in an asynchronous system consisting
of $n$ processes with at most $f$ Byzantine faults. Although these
prior papers dealt with Byzantine faults, it turns out that
their proof of lower bound on $n$ (i.e., lower bound of $(d+2)f+1$)
is also directly applicable to approximate {\em vector} consensus under the
{\em crash fault with incorrect inputs} model used in our
present work.
Thus, $n\geq (d+2)f+1$ is a necessary condition for {\em vector} consensus
under this fault model. Secondly, it is easy to show
that an algorithm for approximate {\em convex} consensus can be transformed
into an algorithm for approximate {\em vector} consensus.
Therefore, $n\geq (d+2)f+1$ is a necessary condition for approximate
{\em convex} consensus as well. For brevity, we omit a formal proof
of the lower bound, and our
subsequent discussion under the {\em crash faults with incorrect inputs}
model assumes that
\begin{eqnarray}
 n &\geq& (d+2)f+1\label{e_n_bound}
\end{eqnarray}
Our algorithm is correct under this condition, and thus achieves
optimal fault resilience.
For crash faults with {\em correct inputs}, a smaller $n$
suffices, as discussed later in Section \ref{s_crash_good_inputs}.
% Since $d\geq 1$, it follows that $n\geq 3f+1$.

% of the $d$-dimensional input vectors at the fault-free processes (Vector consensus), and show that $(d+2)f+1$ is the lower bound on $n$.  The lower bound proof in \cite{Vaidya_BVC} also applies to reaching approximate Vector consensus under crash failures with incorrect inputs.\footnote{In the necessity proof of \cite{Vaidya_BVC}, the only faulty behavior assumed is in fact crash failure with incorrect inputs.} This implies that $n \geq (d+2)f+1$ is also necessary to solve Convex Consensus under crash failures with incorrect inputs.  We do not reproduce the lower bound proof here, but in the rest of the paper, we assume that $n\geq (d+2)f+1$, and also that $n\geq 2$ and $d \geq 1$ (because consensus is trivial when $n=1$ or $d = 0$). Interestingly, as we will discuss in Section \ref{s_crash}, if we consider only crash failures (without incorrect inputs), then we only need at least $f+1$ processes. +++++ need to verify the statement. +++++
 
\item In this paper, we only consider deterministic algorithms. A
convex consensus algorithm $A$ is said to be optimal if the
following condition is true:

\begin{list}{}{}
\item[]
Let $F$ denote a set of up to $f$ faulty processes.
For a {\bf given execution} of algorithm $A$ with 
$F$ being the set of faulty processes, 
let $y_i(A)$ denote the output polytope at process $i$ at the end
of the given execution.
For any other convex consensus algorithm $B$,
{\bf there exists} an execution
with $F$ being the set of faulty processes, such that 
$y_i(B)$ is the output at fault-free process $i$, and
$y_j(B)\subseteq y_j(A)$ for {\bf each} fault-free process $j$.
\end{list}
The goal here is to decide on an output polytope that includes
as much of the convex hull of {\em all} correct inputs as possible.
However, since any process may be potentially faulty (with incorrect input), the output
polytope can be smaller than the convex hull of {\em all} correct
inputs. Intuitively speaking, the optimality condition says that an optimal
algorithm should decide on a convex region that is
{\em no smaller than that decided in a worst-case execution} of algorithm $B$.
In Section \ref{s_optimal}, we will show that our proposed algorithm
is optimal in the above sense.

% We prove that Algorithm 1 allows the processes to agree on a polytope that is guaranteed to contain a polytope that is named $I_Z$ in later analysis. $I_Z$ is a function of the inputs at some of the fault-free processes under adversarial conditions. We show that, for any Convex Consensus algorithm, there exists an execution (with certain faulty behavior and message delay pattern) in which the fault-free processes must agree on a polytope that is equal to or contained in $I_Z$. Thus, as per Definition \ref{def:optSize}, the output polytope chosen by our algorithm is {\em optimal}.

\end{itemize}

\paragraph{Summary of main contributions of the paper:}
\begin{itemize}
\item The paper introduces the problem of {\em convex
consensus}. We believe that feasibility of
convex consensus can be used to infer feasibility of solving other
interesting problems as well.

% of which allows the fault-free processes to agree on a convex polytope. The Convex Consensus algorithms may be used as primitives of a large range of consensus problems, such as Vector consensus \cite{herlihy_multi-dimension_AA,Vaidya_BVC} and convex function optimization \cite{Boyd_optimization,AA_convergence_markov,Nedic_convex}.

\item We present an approximate convex consensus algorithm in {\em asynchronous} systems,
and show that it achieves optimality in terms of
its resilience, and also in terms of the convex polytope that it
decides on.

\item We show that the convex consensus algorithm can be used to solve a
version of the {\em convex function optimization} problem. 
We also prove an impossibility result pertaining to convex function
optimization with crash faults in asynchronous systems.
\end{itemize}

% \paragraph{Related Work:}

\paragraph{Related Work:}

For brevity, we only discuss the most relevant prior work here.
Many researchers in the decentralized control area, including Bertsekas and Tsitsiklis \cite{AA_convergence_markov}
and Jadbabaei, Lin and Moss \cite{jadbabaie_concensus}, have explored approximate consensus
{\em in the absence of process faults}, using only near-neighbor communication in systems
wherein the communication graph may
be partially connected and time-varying.
The structure of the proof of correctness of the algorithm presented in this paper,
and our use of well-known matrix analysis results \cite{Wolfowitz},
is inspired by the above prior work.
We have also used similar proof structures in our prior work on other (Byzantine) consensus
algorithms \cite{Tseng_general,Vaidya_BVC}.
With regards to the proof technique, this paper's contribution is to show
how the above proof structure can be extended to the case when the process
state consists of convex polytopes. 

Dolev et al. addressed approximate Byzantine consensus in both synchronous and asynchronous systems \cite{AA_Dolev_1986} (with scalar input). Subsequently, Coan proposed a simulation technique to transform consensus algorithms that are resilient to crash faults into algorithms tolerating Byzantine faults  \cite{Coan_Simulation,Welch_textbook}. Independently, Abraham,
Amit and Dolev proposed an algorithm for approximate Byzantine consensus \cite{abraham_04_3t+1_async}.
As noted earlier,
the recent work of Mendes and Herlihy \cite{herlihy_multi-dimension_AA}
and Vaidya and Garg \cite{Vaidya_BVC} has addressed approximate {\em vector} consensus in the presence of Byzantine faults.
This work has
yielded lower bounds on the number of processes, and algorithms with optimal resilience for asynchronous \cite{herlihy_multi-dimension_AA,Vaidya_BVC} as well as synchronous systems \cite{Vaidya_BVC} modeled as complete graphs.
Subsequent work \cite{Vaidya_incomplete} has explored the vector consensus problem in incomplete graphs.

%  All three works \cite{herlihy_multi-dimension_AA,Vaidya_BVC,Vaidya_incomplete} considered only the consensus problem where the decision is a single point in the $d$-dimensional space. Moreover, it is not clear whether the proposed algorithms can be generalized to solve Convex Consensus problem, since the algorithms do not guarantee any properties on the decision points at fault-free processes (e.g., location of the decision point), except that the decision points are in the convex hull of the input vectors at the fault-free processes. On the contrary, approximate Convex Consensus (ACC) algorithms can be used to solve approximate vector consensus. For example, processes can use a pre-determined rule to decide on a point (say the centroid) of the convex polytopes returned by any ACC algorithms.

Mendes, Tasson and Herlihy \cite{herlihy_colorless_async} study the problem of {\em Barycentric} agreement. {\em Barycentric} agreement has some similarity to convex consensus, in that the output of Barycentric agreement is not
limited to a single value (or a single point). 
However, the correctness conditions for Barycentric agreement are different from
those of our convex consensus problem.

% ++++++++++++++ DISCUSS WITH ME +++++++++++ ++++++++++++ NEED TO EDIT STILL ++++++++++++++++++++++++++
%  (i) the inputs for BA are vertices of a simplex, while the inputs for CC are arbitrary points; (ii) the outputs for BA are sets of vertices of a simplex, while the inputs for CC are convex polytopes (sets of points); and (iii) the agreement property for BA requires containment (i.e., sets of output vertices at fault-free processes are totally ordered by containment), while the agreement property for CC adopts the ``approximate'' notion (i.e., the distance between the outputs at any pair of fault-free processes is within $\epsilon$.). Due to these differences, it is not clear whether the BA algorithm in \cite{herlihy_colorless_async} can be adapted to solve CC.
%Mendes et al. generalized the application of the topological model to colorless tasks in asynchronous systems with Byzantine failures . 

% However, this paper is the first to develop transition matrix representation for an algorithm in which the processes' states are convex polytopes.

%\input{primitive}

\section{Preliminaries}
\label{s_ops}

Some notations introduced in the paper are summarized in Appendix \ref{app_s_notations}.
In this section,
we introduce functions $\HH$, $\LL$, and a communication primitive used in our algorithm.

\begin{definition}
\label{def:hh}
% Function $\HH$:
For a multiset of points $X$, $\HH(X)$ is 
the convex hull of the points in $X$.
\end{definition}
A multiset may contain the same element more than once.
\begin{definition}
\label{def:linear_hull}
Function ${\normalfont\LL}$:
Suppose that $\nu$ non-empty convex polytopes $h_1, h_2,\cdots, h_\nu$, and $\nu$ weights $c_1, c_2, \cdots, c_\nu$
are given such that
$0 \leq c_i \leq 1$ and $\sum_{i = 1}^\nu c_i = 1$,
Linear combination of these convex
polytopes, ${\normalfont \LL}([h_1, h_2,\cdots, h_\nu]~;~[c_1, c_2, \cdots, c_\nu])$,
 is defined as follows:
\begin{itemize}
\item 
$p \in {\normalfont \LL}([h_1, h_2,\cdots, h_\nu];[~c_1, c_2,\cdots, c_\nu])$ if and only if
\begin{equation}
\label{eq:linear_hull}
\text{for~} 1\leq i\leq \nu,
\text{~there exists~} p_i \in h_i, ~~\text{such that}~~p = \sum_{1\leq i\leq \nu} c_i p_i
\end{equation}
\end{itemize}
\comment{+++++++++++++++++++++++++++++++++
Function ${\normalfont\LL}$:
Suppose that $\nu$ convex polytopes $h_1, h_2,\cdots, h_\nu$, and $\nu$ constants $c_1, c_2, \cdots, c_\nu$
are given such that
(i) $0 \leq c_i \leq 1$ and $\sum_{i = 1}^\nu c_i = 1$,
and (ii) for $1\leq i\leq\nu$, if $c_i\neq 0$, then $h_i\neq\emptyset$.
Linear combination of these convex
polytopes, ${\normalfont \LL}(h_1, h_2,\cdots, h_\nu;~c_1, c_2, \cdots, c_\nu)$,
 is defined as follows:
\begin{itemize}
\item Let $Q := \{ i ~|~ c_i\neq 0, ~ 1\leq i\leq \nu \}$.
\item 
$p \in {\normalfont \LL}(h_1, h_2,\cdots, h_\nu;~c_1, c_2,\cdots, c_\nu)$ if and only if
\begin{equation}
\label{eq:linear_hull}
\text{for each~} i\in Q,
\text{~there exists~} p_i \in h_i, ~~\text{such that}~~p = \sum_{i \in Q} c_i p_i
\end{equation}
\end{itemize}
++++++++++++++++++++++++++++++++++++}
\end{definition}
% Note that a convex polytope may possibly consist of a single point.
Because $h_i$'s above are all convex and non-empty, $\LL([h_1, h_2,\cdots, h_\nu]~;~[c_1, c_2,\cdots, c_\nu])$ is also a
convex non-empty polytope.
 (The proof is straightforward.)
%  included in Lemma \ref{lemma:linear_valid} in Appendix \ref{a_lemmas}.)
The parameters for $\LL$ consist of two vectors, with the elements
of the first vector being polytopes, and the elements of the second vector
being the corresponding weights in the linear combination.
With a slight abuse of notation, we will 
also specify the vector of polytopes as a multiset --
in such cases,
we will always assign an identical weight to all the polytopes in the multiset,
and hence their ordering is not important.

\paragraph{{\em Stable vector} communication primitive:}
As seen later, our algorithm proceeds in asynchronous rounds. In round 0 of the algorithm, the processes
use a communication primitive called {\em stable vector} \cite{renaming_JACM, herlihy_colorless_async}, to
try to
learn each other's inputs. {\em Stable vector} was originally developed in the context
of Byzantine faults \cite{renaming_JACM,herlihy_colorless_async}.
To achieve its desirable properties (listed below), {\em stable vector} requires
at least $3f+1$ processes, with at most $f$ being Byzantine faulty.
Since the {\em crash fault with incorrect inputs} model is weaker than the Byzantine
fault model, the properties of {\em stable vector} listed below will hold in our context,
provided that $n\geq 3f+1$.
As noted earlier in Section \ref{s_intro},
$n\geq (d+2)f+1$ is a necessary condition for approximate convex consensus in the presence of crash faults
with incorrect inputs. Then, with $d\geq 1$, we have $n\geq 3f+1$, and the properties of stable vector below will
hold.

In round 0 of our algorithm, each process $i$ first broadcasts a message
consisting of the tuple $(x_i,i,0)$, where $x_i$ is process
$i$'s input. In this tuple, 0 indicates the (asynchronous) round index. Process $i$ then waits for the {\em stable vector} primitive
to return a set $R_i$ containing round 0 messages.
We will rely on the following properties of the {\em stable vector} primitive,
% to prove the correctness (Section \ref{s_correctness}) and optimality (Section \ref{s_size}) of the proposed algorithm.
which are implied by results proved in prior work \cite{renaming_JACM,herlihy_colorless_async}.
\begin{itemize}
\item {\bf Liveness}:
At each process $i$ that does not crash before the
end of round 0,
{\em stable vector} returns a set $R_i$ containing at least
$n-f$ distinct tuples of the form $(x,k,0)$.

\item {\bf Containment}: For processes $i,j$ that do not crash
before the end of round 0, let $R_i, R_j$ be the set of messages returned to processes $i, j$ by {\em stable vector} in round 0, respectively. Then,
either $R_i \subseteq R_j$ or $R_j \subseteq R_i$.
(Also, by the previous property, $|R_i|\geq n-f$ and $|R_j|\geq n-f$.)
\end{itemize}
A description of the implementation of the {\em stable vector} primitive is omitted for
lack of space. Please refer to \cite{renaming_JACM,herlihy_colorless_async} for more details.  

%++++++++++++++++++ how about the other properties of stable vector? ++++++++++++++++++++

% Due to lack of space, we omit other properties and implementation of {\em stable vector}. Please refer to \cite{renaming_JACM,herlihy_colorless_async} for more details. Our technical report \cite{Tseng_BCC} also shows how to integrate stable vector with Convex Consensus algorithms.

%\input{alg}

\section{Proposed Algorithm and its Correctness}
\label{s_alg}

The proposed algorithm, named {\em Algorithm CC}, proceeds in asynchronous rounds.
The input at each process $i$ is named $x_i$. 
The initial round of the algorithm is called round $0$.
Subsequent rounds are named round 1, 2, 3, etc.
In each round $t\geq 0$, each process $i$ computes a state variable $h_i$, which represents a convex polytope in the $d$-dimensional Euclidean space. We will refer to the value of $h_i$ at the {\em end} of the $t$-th round performed by process $i$ as $h_i[t]$, $t\geq 0$. Thus, for $t\geq 1$, $h_i[t-1]$ is the value of $h_i$ at the {\em start} of the $t$-th round at process $i$. The algorithm terminates after $t_{end}$ rounds, where $t_{end}$ is a constant defined later in equation (\ref{e_end}). 
The state $h_i[t_{end}]$ of each fault-free process $i$
at the end of $t_{end}$ rounds is its output (or decision) for the
consensus algorithm.

$X_i$ and $Y_i[t]$ defined on lines 4 and 13 of the algorithm are both {\em multisets}.
A given value may occur multiple times in a multiset.
Also, the intersection in line 5 is over the convex hulls of
the subsets of multiset $X_i$ of size $|X_i| - f$ (note that
each of these subsets is also a multiset).
Elements of $X_i$ are points in the $d$-dimensional Euclidean space,
whereas elements of $Y_i[t]$ are convex polytopes.
In line 14, $Y_i[t]$ specifies the multiset of
polytopes whose linear combination is obtained
using $\LL$; all the weights specified as parameters to $\LL$ here are equal to $\frac{1}{|Y_i[t]|}$. 

\vspace*{8pt}
\hrule

\vspace*{2pt}

\noindent {\bf Algorithm CC:} Steps performed at process $i$ shown below.
% \\~\\The algorithm terminates after $t_{end}$ rounds

\vspace*{4pt}

\hrule

\vspace*{8pt}

\noindent {\bf Initialization:} All sets used below are initialized to $\emptyset$.\\

%\begin{enumerate}
%\item Round 0: 

\noindent {\bf Round $0$ at process $i$:} 	

\begin{itemize}
\item On entering round 0: \hfill 1

\hspace*{0.8in} 
Send message $(x_i,i,0)$ to all the processes \hfill 2

\item When {\em stable vector} returns a set $R_i$: \hfill 3
% As noted in Section \ref{s_ops}, $R_i$ contains at least $n-f$ tuples of the form $(x,j,0)$. \hfill 2

\hspace*{0.8in} Multiset $X_i := \{\,x \,|\,(x,k,0)\in R_i\}$ \hfill // Note: $|X_i|=|R_i|$ \hfill 4

% \hspace*{2in}// Note that $|X_i|=|R_i|$. \\
% \hspace*{2in}// Identical element may appear multiple times in a multiset.

\hspace{0.8in}
$h_i[0]~:= ~ \cap_{\,C \subseteq X_i,\, |C| = |X_i| - f}~~\HH(C)$ \hfill 5

% \hspace*{2in} // Intersection above is over the convex hulls of\\
% \hspace*{2in} // the subsets of $X_i$ of size $|X_i| - f$.

\hspace*{0.8in} Proceed to Round 1 \hfill 6

\end{itemize}

\noindent {\bf Round $t\geq 1$ at process $i$:} 

\begin{itemize}
\item On entering round $t\geq 1$: \hfill 7

\hspace*{0.8in} $\msg_i[t] := \msg_i[t] \cup (h_i[t-1],i,t)$ \hfill 8

\hspace*{0.8in} Send message $( h_i[t-1], i, t)$ to all the processes \hfill 9 

\item When message $(h, j, t)$ is received from process $j\neq i$  \hfill 10

    \hspace*{0.8in} $\msg_i[t] := \msg_i[t] \cup \{(h, j, t)\}$ \hfill 11

\item When $|\msg_i[t]|\geq n-f$ for the first time: \hfill 12

% \hspace*{0.8in} $h_i[t] := H(\msg_i[t],t)$  \hfill 9

\hspace*{0.8in} Multiset $Y_i[t]:=\{h~|~ (h,j,t)\in \msg_i[t]\}$
	 \hfill // Note: $|Y_i[t]|=|\msg_i[t]|$ \hfill 13

% \hspace*{2in} // It turns out that $h\in Y_i[t]$ is always non-empty.

\hspace*{0.8in} $h_i[t]~ :=~ \LL(~~Y_i[t]~;~[\frac{1}{|Y_i[t]|}, \cdots,\frac{1}{|Y_i[t]|}])$
\hfill 14

% \hspace*{2in} // $Y_i[t]$ above is a multiset of convex polytopes, and\\
% \hspace*{2in} // all the weights in this use of $\LL$ are equal to $\frac{1}{|Y_i[t]|}$. 

\hspace*{0.8in} If $t<t_{end}$, then proceed to Round $t+1$ \hfill 15
\end{itemize}
\hrule

~

\newcommand{\printcomment}[1]{{\bf #1}}

\subsection{Proof of Correctness}
\label{ss_correctness}

The use of matrix representation in our correctness proof below is inspired by the prior work on non-fault-tolerant consensus (e.g., \cite{jadbabaie_concensus,AA_convergence_markov}).
We have also used such a proof structure in our work on Byzantine consensus \cite{Tseng_general,Vaidya_BVC}.
We now introduce more notations (some of the notations are summarized in Appendix \ref{app_s_notations}):
\begin{itemize}
\item For a {\em given} execution of the proposed algorithm, let $F$ denote the {\em actual} set of faulty processes in
that execution. Processes in $F$ may have incorrect inputs, and they may potentially crash.
% Let $|F|=\phi$. Thus, $0\leq \phi\leq f$.

\item For round $r \geq 0$, let $\fbar[r]$ denote the set of faulty processes that have crashed before sending any round $r$ messages. Note that $\fbar[r]\subseteq \fbar[r+1] \subseteq F$.

%\item For round $r \geq 0$, define $\f[r] = F-\fbar[r]$, i.e., the round-$r$ message from processes in $\f[r]$ will eventually be received by some process.

%\item For $r\geq 0$, let $F_v[r]$ denote the set of faulty processes that are still alive in round $r+1$, i.e., their round $r+1$ messages are eventually received by all the fault-free processes. Note that $F_v[r]\subseteq F$.

\end{itemize}
Proofs of Lemmas \ref{l_progress} and \ref{lemma:J_in_H0}
below are presented in Appendices \ref{a_lemma_progress}, and \ref{a_lemma_J}, respectively.
% These lemmas are used to prove the correctness of the {\em Algorithm CC}.
%
\begin{lemma}
\label{l_progress}
Algorithm CC ensures {\em progress}: (i) all the fault-free processes will eventually progress to round 1; and, (ii)
if all the fault-free processes progress to the start of round $t$, $t \geq 1$, then all the
fault-free processes will eventually progress to the start of round $t+1$.
\end{lemma}

\begin{lemma}
\label{lemma:J_in_H0}
For each process $i\in V-\fbar[1]$, the polytope $h_i[0]$ is non-empty and convex.
\end{lemma}

\comment{+++++++++++++++++++++

 Then, we develop a
{\em transition matrix} representation of the proposed algorithm, and use
that to prove its correctness. The proof below follows a similar structure in our prior work on different consensus algorithms in incomplete networks \cite{Tseng_general,Vaidya_incomplete,Tseng_link}; however, such analysis has not been applied in the case of Convex Consensus.

+++++++++++++++++++++}

% \subsection{Matrix Preliminaries}

We now introduce some matrix notation and terminology to be used in our proof.
Boldface upper case letters are used below to denote matrices, rows of matrices, and their elements. For instance, $\bfA$ denotes a matrix, $\bfA_i$ denotes the $i$-th row of matrix $\bfA$, and $\bfA_{ij}$ denotes the element at the intersection of the $i$-th row and the $j$-th column of matrix $\bfA$.
A vector is said to be {\em stochastic} if all its elements
are non-negative, and the elements add up to 1.
A matrix is said to be row stochastic if each row of the matrix is a
stochastic vector \cite{jadbabaie_concensus}. 
For matrix products, we adopt the ``backward'' product convention below, where $a \leq b$,
\begin{equation}
\label{backward}
\Pi_{\tau=a}^b \bfA[\tau] = \bfA[b]\bfA[b-1]\cdots\bfA[a]
\end{equation}
Let $\vectorv$ be a column vector of size $n$ whose elements are convex polytopes. The $i$-th element of $\vectorv$ is $\vectorv_i$. Let $\bfA$ be a $n\times n$ row stochastic square matrix.
 We define the product of $\bfA_i$ (the $i$-th row of $\bfA$)
and $\vectorv$ as follows using function $\LL$ defined 
% Definition \ref{def:linear_hull}
 in Section \ref{s_ops}.
\begin{eqnarray}
\label{e_r_c}
\bfA_i\vectorv &=& \LL(\vectorv^T\,;~\bfA_i)
\end{eqnarray}
where $^T$ denotes the transpose operation.
The above product is a polytope in the $d$-dimensional Euclidean space.
Product of matrix $\bfA$ and $\vectorv$ is then defined as follows:
\begin{equation}
\label{eq:multiplication}
\bfA \vectorv = [~\bfA_1 \vectorv
~~~~~\bfA_2\vectorv
~~~~~\cdots~~~~~
\bfA_n\vectorv~]^T
\end{equation}
Due to the transpose operation above, the product $\bfA\vectorv$ is a column vector
consisting of $n$ polytopes. Now, we present a useful lemma. The lemma is proved in Appendix \ref{a_matrix}

\begin{lemma}
	\label{lemma:matrix_comp}
	For two $n$-by-$n$ matrices $\bfA$ and $\bfB$, and an $n$-element column vector of $d$-dimensional polytopes $\vectorv$, we have $\bfA (\bfB \vectorv) = (\bfA \bfB) \vectorv$.
\end{lemma}
 
Now, we describe how to represent Algorithm CC using a matrix form. Let $\vectorv[t]$, $t\geq 0$, denote a column vector of length $n$.
In the remaining discussion,
we will refer to $\vectorv[t]$ as the state of the system at the end of round $t$.
In particular,
$\vectorv_i[t]$ for $i\in V$ is viewed as
the state of process $i$ at the end of round $t$.
We define $\vectorv[0]$ as follows as {\em initialization} of the state vector:
\begin{itemize}
\item[(I1)] For each process $i\in V-\fbar[1]$, $\vectorv_i[0]:=h_i[0]$.
\item[(I2)] Pick any one fault-free process $m\in V-F\subseteq V-\fbar[1]$. For each process $k\in \fbar[1]$,
$\vectorv_k[0]$ is {\em arbitrarily} defined to be equal to $h_m[0]$.
Such an arbitrary choice suffices because the state
$\vectorv_k[0]$ for $k\in \fbar[1]$ does not impact future state
of any other process (by definition, processes in $\fbar[1]$
do not send any messages in round 1 and beyond).
\end{itemize}

We will show that the state evolution can be expressed using
matrix form as in (\ref{matrix:alg1}) below, where
$\matrixm[t]$ is an $n\times n$ matrix with certain desirable properties.
The state $\vectorv_k[t]$ of process $k\in \fbar[t]$ is not
meaningful, since process $k$ has crashed.
However, (\ref{matrix:alg1}) assigns it a value
for convenience of analysis.
$\matrixm[t]$ is said to be the {\em transition matrix} for round $t$.
\begin{equation}
\label{matrix:alg1}
\vectorv[t] = \matrixm[t]~\vectorv[t-1], ~~~~~ t\geq 1
\end{equation}

In particular,
given an execution of the algorithm, we construct the transition matrix
$\bfM[t]$
for round $t\geq 1$ of that execution using the two rules below ({\em Rule 1} and {\em Rule 2}).
Elements of row $\bfM_i[t]$ will determine the state $\vectorv_i[t]$ of process $i$
(specifically, $\vectorv_i[t] = \bfM_i[t]\vectorv[t-1]$).
Note that {\em Rule 1} applies to processes in $V-\fbar[t+1]$. Each process
$i\in V- \fbar[t+1]$
survives at least until the start of round $t+1$, and sends at least one message in round $t+1$.
Therefore, its state $\vectorv_i[t]$ at the end of round $t$ is of consequence. 
On the other hand, processes in $\fbar[t+1]$ crash sometime before 
sending any messages in round $t+1$ (possibly crashing in previous rounds). Thus, their states at the end of round $t$ are not relevant
to the fault-free processes anymore, and hence {\em Rule 2} defines the entries of the corresponding
rows of $\bfM[t]$ somewhat arbitrarily.

\vspace*{8pt}

\hrule 

In the matrix specification below,
$\msg_i[t]$ is the message set at the point where $Y_i[t]$
is defined on line 13 of the algorithm.
Thus, $Y_i[t]:=\{h~|~ (h,j,t)\in \msg_i[t]\}$, and $|\msg_i[t]|=|Y_i[t]|$.

\begin{itemize}
\item {\em Rule 1}: For each process $i \in V-\fbar[t+1]$, and each $k\in V$:

\begin{list}{}{}
\item{} If a round $t$ message from process $k$  (of the form $(*,k,t)$) is in $\msg_i[t]$, then 

\begin{equation}
\label{eq:matrix_i}
\matrixm_{ik}[t]  :=  \frac{1}{|\msg_i[t]|}
\end{equation}

\item{}
Otherwise,
\begin{equation}
\label{eq:matrix_i-2}
\matrixm_{ik}[t]  :=  0
\end{equation}

\end{list}

\item {\em Rule 2}: For each process $j \in \fbar[t+1]$, 
and each $k \in V$,
\begin{eqnarray}
\matrixm_{jk}[t] &:=& \frac{1}{n}
\label{e_fv}
\end{eqnarray}

\end{itemize}
\hrule

~

\noindent
% For all $t\geq 0$, Theorem \ref{t_M} below proves that, for each $i\in V-\fbar[t+1]$, $h_i[t]=\vectorv_i[t]$.

\begin{theorem}
\label{t_M}
For $t\geq 1$,
define $\vectorv[t]=\bfM[t]\vectorv[t-1]$, with $\bfM[t]$ as specified above.
Then, for $\tau\geq 0$, and for all $i\in V-\fbar[\tau+1]$,
% $h_i[t]$ is non-empty, and
$\vectorv_i[\tau]$ equals $h_i[\tau]$.
\end{theorem}

The proof is presented in Appendix \ref{a_t_M}. The above theorem states that, for $t\geq 1$, equation (\ref{matrix:alg1}), that is,
$\vectorv[t]=\matrixm[t]\vectorv[t-1]$, correctly characterizes the state of the processes
that have not crashed before the end of round $t$.
For processes that have crashed, their states are not relevant, and could be assigned
any arbitrary value for analytical purposes (this is what {\em Rule 2} above effectively does).
Given the matrix product definition in (\ref{eq:multiplication}), and by repeated
application of the state evolution equation (\ref{matrix:alg1}) and Lemma \ref{lemma:matrix_comp},
we obtain
% \[ \bfM[\tau+1] ~ \left(\bfM[\tau] \vectorv[\tau-1]\right) ~=~ \left(\bfM[\tau+1]\bfM[\tau]\right)~\vectorv[\tau-1]
% \text{~ for~} \tau \geq 1.
% \]
% Therefore, by repeated application of (\ref{matrix:alg1}),
% we obtain:
\begin{eqnarray}
\vectorv[t] & = & \left(\,\Pi_{\tau=1}^t \matrixm[\tau]\,\right)\, \vectorv[0],
 ~~~~ t\geq 1
\label{e_unroll}
\end{eqnarray}
Recall that we adopt the ``backward'' matrix product convention presented in
(\ref{backward}).
% Based on the above definition of $\bfM[t]$, the following properties can be proved.

\begin{definition}
\label{def:valid_hull}
A polytope is {\em valid} if it is contained in the convex hull of the inputs of 
fault-free processes. 
\end{definition}

\begin{theorem}
\label{thm:correctness}
% Algorithm CC is correct.
Algorithm CC satisfies the {\em validity}, {\em $\epsilon$-agreement} and {\em termination} properties.
% after a large enough number of asynchronous rounds.
\end{theorem}

\begin{proofSketch}
Appendix \ref{a_correctness} presents the complete proof.
Repeated application of
Lemma \ref{l_progress} ensures that the fault-free processes will progress
to the end of round $t$, for $t\geq 1$.
By repeated application of Theorem \ref{t_M},
$h_i[t]$ equals the $i$-th element of $(\Pi_{\tau=1}^t \bfM[\tau])\vectorv[0]$, for
$i\in V-\fbar[t+1]$.

\noindent{\normalfont\em Validity:}
By design, $\bfM[\tau]$ is a row stochastic matrix for each $\tau$,
therefore, $\Pi_{\tau=1}^t \bfM[\tau]$ is also row stochastic.
As shown in Lemma \ref{lemma:valid_initial_hull} in Appendix \ref{a_lemmas},
 $\vectorv_i[0]=h_i[0]$ is valid for each fault-free process $i\in V-\fbar[1]$.
Also, for each $k\in \fbar[1]$, in initialization step (I2), we defined $\vectorv_k[0]=h_m[0]$,
where $m$ is a fault-free process. Hence, $\vectorv_k[0]$ is valid for
process $k\in \fbar[1]$. Therefore, all the elements of $\vectorv[0]$ are valid.
This observation, in conjunction with the previous observation that
$\Pi_{\tau=1}^t \bfM[\tau]$ is row stochastic,
and the product definition in 
(\ref{eq:multiplication}), implies that
each element of
$\vectorv[t] = (\Pi_{\tau=1}^\tau \bfM[\tau])\vectorv[0]$
is also valid.
Then, Theorem \ref{t_M}  implies that the state of each fault-free process
is always valid, and hence its output (i.e., its state after $t_{end}$
rounds) meets the validity condition.

\noindent{\normalfont \em $\epsilon$-Agreement and Termination:}
To simplify the termination of the algorithm, we assume that the input at each process
belongs to a bounded space; in particular, each coordinate of $x_i$
is lower bounded by $\mu$ and upper bounded by $U$, where 
$\mu$ and $U$ are known constants.
Let $\bfP[t] = \Pi_{\tau=1}^t\, \bfM[\tau]$.
Then, as shown in Lemma \ref{lemma:transition_matrix2} (Appendix \ref{a_lemmas}),
for $i,j\in V-F$ and $k\in V$,
\begin{eqnarray}
\label{e_P}
\|\, \bfP_{ik}[t] -  \bfP_{jk}[t]\,\| & \leq&  \left(1-\frac{1}{n}\right)^t
\end{eqnarray}
where $\|x\|$ denotes absolute value of a real number $x$.
Recall from previous discussion that, due to Theorem \ref{t_M}, for each fault-free process $i$,
$h_i[t]$ equals the $i$-th element of $\bfP[t]\vectorv[0]$.
This in conjunction with (\ref{e_P}) can be used to prove that
for $i,j\in V-F$, the Hausdorff distance between $h_i[t]$
and $h_j[t]$ is bounded.
In particular, for $i,j\in V-F$,
\[
\D(h_i[t],h_j[t]) <
\left(1-\frac{1}{n}\right)^t~\sqrt{dn^2\max(U^2,\mu^2)}
\]
By defining $t_{end}$ to be the smallest integer satisfying the
inequality below,
$\epsilon$-agreement and termination conditions both follow.
\begin{equation}
\label{e_end}
\left(1-\frac{1}{n}\right)^t~\sqrt{dn^2\max(U^2,\mu^2)} < \epsilon
\end{equation}
\end{proofSketch}

\subsection{Optimality of Algorithm CC}
\label{s_optimal}

Due to the {\em Containment} property of {\em stable vector} mentioned in Section \ref{s_ops},
the set $Z$ defined below contains at least $n-f$ messages. Recall that set $R_i$ is defined
on line 3 of Algorithm CC.
\begin{eqnarray}
Z &:=& \cap_{i\in V-F}~ R_i
\label{e_Z}
\end{eqnarray}
Define multiset $X_Z := \{x~|~(x, k, 0) \in Z\}$. Then, define a convex polytope $I_Z$ as follows.
\begin{eqnarray}
I_Z& := &\cap_{D \subset X_Z, |D| = |X_Z|-f}\, \HH(D)
\label{e_I_Z}
\end{eqnarray}

Now we establish a ``lower bound'' on output at the fault-free processes.
% Recall that $\fbar[1]$ is defined as all the faulty processes whose round 1 message is not received by any process. The proof is presented in Appendix \ref{a_l:svSize}.

\begin{lemma}
\label{lemma:svSize}
For all $i\in V-\fbar[t+1]$ and $t \geq 0$,
 $I_Z \subseteq h_i[t]$.
\end{lemma}

Lemma \ref{lemma:svSize} is proved in Appendix \ref{a_l:svSize}.
The following theorem is proved in Appendix \ref{a_t:optSize}.

\begin{theorem}
\label{thm:optSize}
{\em Algorithm CC} is optimal under the notion of optimality in Section \ref{s_intro}.
\end{theorem}

% \section{Discussion}

\paragraph{Degenerate Cases:}
In some cases, the output polytope at fault-free processes may be a single
point, making the output equivalent to that obtained from
{\em vector consensus} \cite{herlihy_multi-dimension_AA,Vaidya_BVC}.
As a trivial example, this occurs when all the fault-free processes have
identical input.
It is possible to identify scenarios when the number of
processes is exactly equal to the lower bound, i.e.,
$n=(d+2)f+1$ processes, when the output polytope consists of just a single
point. However, in general, particularly when $n$ is larger than the
lower bound, the output polytopes will contain infinite number of points.
In any event, as shown in Theorem \ref{thm:optSize}, our algorithm achieves optimality in all cases.
Thus, any other algorithm can also produce such degenerate outputs for the
same inputs.

\subsection{Convex Consensus under Crash Faults with Correct Inputs}
\label{s_crash_good_inputs}

With some simple changes, our algorithm and results can be extended
to achieve convex consensus under the {\em crash faults with correct inputs}
model. Under this model, we still need to 
satisfy the $\epsilon$-agreement and termination properties stated in Section \ref{s_intro}. The
validity property remains unchanged as well, however, in this model, inputs
at all processes are always correct. Thus, validity implies that the ouput
will be contained in the convex hull of the inputs at {\em all} the processes.

To obtain the algorithm for convex consensus under the {\em crash faults with correct inputs} model,
three key changes required. First, the lower bound
on the number of processes becomes $n\geq 2f+1$, which is independent of the
dimension $d$. Second, we need a version of the {\em stable vector}
primitive that satisfies the properties stated previously
with just $2f+1$ processes (this is feasible).
Finally, instead of the computation in line 5 of Algorithm CC, the computation
of $h_i[0]$ needs to be modified as $h_i[0]:=\HH(X_i)$, where $X_i := \{\,x \,|\,(x,k,0)\in R_i\}$. With these
changes, the modified
algorithm achieves convex consensus under the crash faults with
{\em correct} input model, with the rest of the proof being similar to
the proof for the crash faults with {\em incorrect} inputs model.
The modified algorithm exhibits optimal resilence as well.

\section{Convex Function Optimization}
\label{s_optimization}

A motivation behind our work on {\em convex consensus} was to develop an algorithm
that may be used to solve a broader range of problems. For instance, {\em vector
consensus} can be achieved by first solving {\em convex consensus}, and then
using the centroid of the output polytope of convex consensus as the output of
{\em vector consensus}.
Similarly,
{\em convex consensus} can be used to solve a {\em convex
function optimization} problem \cite{Boyd_optimization,AA_convergence_markov,Nedic_convex}
under the crash faults with {\em incorrect inputs} model.
% The goal is to attempt to optimize a cost function over a domain defined
% as the convex hull of the correct inputs at fault-free processes.
We present an algorithm for this, and then discuss some of its
properties, followed by 
an impossibility result of more general interest.
The desired outcome of the function optimization problem
is
to minimize a cost function, say function $c$, over a domain
consisting of the convex hull of the correct inputs.
The proposed algorithm has two simple steps:
\begin{itemize}
\item Step 1: First solve convex consensus with parameter $\epsilon$. Let $h_i$ be the output polytope
of convex consensus at process $i$.
\item Step 2: The output of function optimization is 
the tuple $(y_i,c(y_i))$, where $y_i= \arg\min_{x\in h_i}\,c(x)$.
\end{itemize}
We assume the following
property for some constant $B$: for any inputs $x,y$, $\|c(x)-c(y)\|\leq B\, \ddd(x,y)$ ($B$-Lipschitz continuity).
Then, it follows that,
for fault-free processes $i,j$, $\| c(y_i)-c(y_j)\|\leq \epsilon B$.
Thus, the fault-free processes find approximately
equal minimum value for the function. However,
$c(y_i)$ at process $i$ may not be minimum over the {\em entire}
convex hull of the inputs of fault-free processes. For instance,
even when all the processes are fault-free, each subset of $f$ processes
is viewed as {\em possibly} faulty with incorrect inputs.
The natural question then is ``Is it possible to find
an algorithm that always find a smaller minimum value than the above
algorithms?''
We can extend the notion of optimality from Section \ref{s_intro}
to function optimization in a natural way, and show that no other algorithm
can obtain a lower minimum value (in the worst-case) than the above algorithm.
Appendix \ref{a_optimization} elaborates on this.

The above discussion implies that for any $\beta>0$, we can achieve
$\| c(y_i)-c(y_j)\|< \beta$ by choosing $\epsilon = \beta/B$ for
convex consensus in Step 1.
However, we are not able to similarly show that $\ddd(y_i,y_j)$ is small.
In particular, if there are multiple points on  the {\em boundary} of
$h_i$  that minimize the cost function, then one of the points is chose
arbitrarily as $\arg\min_{x\in h_i}c(x)$, and consensus on the point is not guaranteed. It turns out it is not feasible to simultaneously reach
(approximate) consensus on a point, and to also ensure that the
cost function at that point is ``small enough''.
We briefly present an impossibility
result that makes a more precise statement of this infeasibility.
It can be shown (Appendix \ref{a_optimization}) that
the following four properties cannot be satisfied simultaneously in the presence of up to $f$ crash faults
with incorrect inputs.\footnote{Similar impossibility result can be shown for the
crash fault with {\em correct} inputs model too.}
The intuition behind part (ii) of the weak optimality condition below is as follows.
When $2f+1$ processes have an identical input, say $x^*$, even if $f$ of them are slow (or crash),
each fault-free process must be able to learn that $f+1$ processes have 
input $x^*$, and at least one of these $f+1$ processes must be fault-free.
Therefore, it would know that the minimum value of the cost function over
the convex hull of the correct inputs is at most $ c(x^*)$. Note that our algorithm above achieves weak $\beta$-optimality but not $\epsilon$-agreement.

\begin{itemize}
\item {\bf Validity}: output $y_i$ at fault-free process $i$ is a
point in the convex hull of the correct inputs.

\item \textbf{$\epsilon$-Agreement}: for any constant $\epsilon > 0$, 
for any fault-free processes
$i,j$, $\ddd(y_i,y_j)<\epsilon$.

\item \textbf{Weak $\beta$-Optimality}:
(i)
for any constant $\beta>0$, for any fault-free processes $i,j$,
$\| c(y_i)-c(y_j)\|<\beta$, and
(ii)
if at least $2f+1$ processes (faulty or fault-free)
have an identical input, say $x$,
then for any fault-free process $i$, $c(y_i) \leq c(x)$.

\item {\bf Termination:} each fault-free process must terminate within a finite amount of time.
\end{itemize}
The proof of impossibility for $n\geq 4f+1$ and $d\geq 1$ is presented in Appendix \ref{a_optimization}.
We know that even without the weak $\beta$-optimality, we need $n\geq (d+2)f+1$. 
Thus, the impossibility result is complete for $d\geq 2$.
Whether the impossibility extends to $3f+1\leq n\leq 4f$ and $d=1$ is presently unknown.

\section{Summary}

We introduce the {\em convex consensus} problem under crash faults with incorrect inputs,
and present an asynchronous approximate convex consensus algorithm
with optimal fault tolerance that reaches consensus on an {\em optimal} output polytope.
We briefly extend the results to the {\em crash faults with correct inputs} model, and also use the convex
consensus algorithm to solve convex function optimization. An impossibility result for
asynchronous function optimization
is also presented.

%\bibliography{paperlist}

\clearpage
\appendix

\section{Notations}
\label{app_s_notations}

This appendix summarizes some of the notations and terminology introduced throughout the paper.

\begin{itemize}
\item $d = $ dimension of the input vector at each process.
\item $n =$ number of processes. We assume that $n\geq (d+2)f+1$.
\item $f =$ maximum number of faulty processes.
\item $\sv = \{1, 2, \cdots, n\}$ is the set of all processes.
%\item $\cdot =$ the dot product operator.
\item $\ddd(p, q) = $ Euclidean distance between points $p$ and $q$.
\item $\D(h_1, h_2) = $ the Hausdorff distance between convex polytopes $h_1, h_2$.
\item $\HH(C) = $ the convex hull of a multiset $C$.
%\item $H_a(\{h_1, h_2, ..., h_k\}) = $ the average of convex polytopes $h_1, h_2, ..., h_k$.
\item $\LL([h_1, h_2, \cdots, h_k];~[c_1, c_2, \cdots, c_k])$, defined in Section \ref{s_ops},
is a linear combination of convex polytopes $h_1, h_2, ..., h_k$ with weights $c_1,c_2,\cdots,c_k$, respectively.
\item $| X | = $ the size of a {\em multiset} or {\em set} $X$.
\item $\| a \| = $ the absolute value of a real number $a$.
\item $F$ denotes the {\em actual} set of faulty processes in an execution of the algorithm.
% \item $\phi=|F|$. Thus, $0\leq \phi\leq f$.
\item $\fbar[t]~~(t\geq 0)$, defined in Section \ref{s_alg}, denotes the set of (faulty) processes
that do not send any messages in round $t$. Thus, each process in $\fbar[t]$ must have
crashed before sending any message in round $t$ (it may have
possibly crashed in an earlier
round).
Note that $\fbar[r]\subseteq \fbar[r+1]\subseteq F$ for $r\geq 1$.
% \item $\alpha=1-\frac{1}{n}$.
\item 
We use boldface upper case letters to denote matrices, rows of matrices, and their elements. For instance, $\bfA$ denotes a matrix, $\bfA_i$ denotes the $i$-th row of matrix $\bfA$, and $\bfA_{ij}$ denotes the element at the intersection of the $i$-th row and the $j$-th column of matrix $\bfA$.

\end{itemize}

\section{Proof of Lemma \ref{l_progress}}
\label{a_lemma_progress}

\noindent
{\bf Lemma \ref{l_progress}:}
{\em
Algorithm CC ensures {\em progress}: (i) all the fault-free processes will eventually progress to round 1; and, (ii)
if all the fault-free processes progress to the start of round $t$, $t \geq 1$, then all the
fault-free processes will eventually progress to the start of round $t+1$.\\
}

\begin{proof}

\paragraph
{\bf Part (i):}

By assumption, all fault-free processes begin the round 0 eventually, and perform a broadcast of their input (line 1).
There at least $3f+1$ processes as argued in Section \ref{s_ops},
and at most $f$ may crash.
The Liveness property of
{\em stable vector} ensures that it will eventually return (on line 3).
Therefore, each process that does not crash in round 0 will eventually proceed to round 1 (line 6).

~
\paragraph
{\bf Part (ii):}

The proof is by induction.
Suppose that the fault-free processes begin round $t\geq 1$. (We already
proved that the fault-free processes begin round 1.) 
Thus, each fault-free process $i$ will perform a broadcast
of $(h_i[t-1],i,t)$ on line 9.
By the assumption of reliable channels,
process $i$ will eventually receive message $(h_j[t-1],j,t)$ from each fault-free process $j$. Thus, it will receive messages from at least $n-f-1$
{\em other} processes, and include these received messages in $\msg_i[t]$ (line 10-11).
Also, it includes (on line 8) its own message into
$\msg_i[t]$. Thus, $\msg_i[t]$ is sure to reach size $n-f$ eventually,
and process $i$ will be able to progress to round $t+1$ (line 12-15).
%Similarly, each fault-free process will eventually proceed to round $t+1$.
\end{proof}

\section{Proof of Lemma \ref{lemma:J_in_H0}}
\label{a_lemma_J}

The proof of Lemma \ref{lemma:J_in_H0} uses the following theorem by Tverberg \cite{tverberg}:
\begin{theorem}
\label{thm:tverberg}
(Tverberg's Theorem \cite{tverberg}) For any integer $f \geq 0$, for every multiset $T$ containing at least $(d+1)f+1$ points in a $d$-dimensional space, there exists a partition
$T_1, .., T_{f+1}$ of $T$ into $f+1$ non-empty multisets such that $\cap_{l=1}^{f+1} \HH(T_l) \neq \emptyset$.
\end{theorem}

~

\noindent
Now, we prove Lemma \ref{lemma:J_in_H0}. \\

\noindent{\bf Lemma \ref{lemma:J_in_H0}:}
{\em
 For each process $i\in V-\fbar[1]$, the polytope $h_i[0]$ is non-empty and convex. \\
}

\begin{proof}

Consider any $i\in V-\fbar[1]$.
Consider the computation of polytope $h_i[0]$ on line 5 of the algorithm as
\begin{equation}
\label{e_h_appendix}
 h_i[0]~:=~ \cap_{\,C \subseteq X_i,\, |C| = |X_i| - f}~~\HH(C),
\end{equation}
 where $X_i := \{\,x \,|\,(x,k,0)\in R_i\}$ (lines 4-5).
Convexity of $h_i[0]$ follows directly from (\ref{e_h_appendix}), because $h_i[0]$ is an
intersection of convex hulls.

Recall that, due to the lower bound on $n$ discussed in Section \ref{s_intro},
we assume that $n \geq (d+2)f+1$.
Thus, $|X_i| \geq n-f \geq (d+1)f+1$. By Theorem \ref{thm:tverberg} above, there exists a partition $T_1, T_2, \cdots, T_{f+1}$ of $X_i$ into multisets ($T_j$'s) such that $\cap_{j=1}^{f+1} \HH(T_j) \neq \emptyset$. Let us define
\begin{equation}
\label{eq:J}
J = \cap_{j=1}^{f+1} \HH(T_j)
\end{equation}
Thus, by Tverberg's theorem above, $J$ is non-empty.
Now, each multiset $C$ used in (\ref{e_h_appendix}) to compute $h_i[0]$
% is of size at least $n-2f$.
% Each $C$
 excludes only $f$ elements of $X_i$, whereas there are $f+1$ multisets
in the partition $T_1,T_2,\cdots,T_{f+1}$ of multiset $X_i$.
Therefore, each multiset $C$ will fully contain at least one multiset
from the partition. It follows that $\HH(C)$
will contain $J$ defined above. Since this property holds true for each multiset
$C$ used to compute $h_i[0]$,
$J$ is contained in
the convex polytope $h_i[0]$ computed as per (\ref{e_h_appendix}).
Since $J$ is non-empty, $h_i[0]$ is non-empty.

\end{proof}

\section{Proof of Lemma \ref{lemma:matrix_comp}}
\label{a_matrix}

Here, we prove Lemma \ref{lemma:matrix_comp}.

\noindent {\bf Lemma \ref{lemma:matrix_comp}} 	{\em For two $n$-by-$n$ matrices $\bfA$ and $\bfB$, and an $n$-element column vector of $d$-dimensional polytopes $\vectorv$, we have $\bfA (\bfB \vectorv) = (\bfA \bfB) \vectorv$.}

~

\begin{proof}
	Let $\vectorl = \bfA (\bfB \vectorv)$ and $\vectorr = (\bfA \bfB) \vectorv$. To prove the lemma, we show that for $1 \leq k \leq n$, $\vectorl_k = \vectorr_k$. %Recall that for each $k, \vectorl_k$ and $\vectorr_k$ are polytopes in the $d$-dimensional space. 
	
	We first show the following claim for an $n$-element column vector of $d$-dimensional points $p$.
	
	\begin{claim}
		\label{claim:matrix_comp}
		$\bfA_k (\bfB p) = (\bfA \bfB)_k p$
	\end{claim}
	
	\begin{proof}
		\begin{align*}
		\bfA_k (\bfB p) &= \sum_{i = 1}^n \bfA_{ki} \left( \sum_{j = 1}^n  \bfB_{ij} p_j \right)\\
		&= \sum_{j = 1}^n \left( \sum_{i=1}^n \bfA_{ki} \bfB_{ij} \right) p_j \\
		&= (\bfA \bfB)_k p
		\end{align*}
	\end{proof}

	\begin{claim}
		\label{claim:rinl}
		For any integer $k$ such that $1 \leq k \leq n$, a point $p \in \vectorr_k$, then $p \in \vectorl_k$.
	\end{claim}
	
	\begin{proof}
		By the definition of $p$, there exists a vector of $d$-dimensional points $q$ such that (i) $(\bfA \bfB)_k q = p$; and (ii) for $1 \leq j \leq n$, $q_j \in \vectorv_j$. Observe that $p = (\bfA \bfB)_k q = \bfA_k (\bfB q)$ due to Claim \ref{claim:matrix_comp}. This implies that $p \in \vectorl_k$ due to our definition of matrix operation over polytopes.
	\end{proof}
	
	Claim \ref{claim:rinl} implies that
	
	\begin{equation}
	\label{eq:rinl}
	\vectorr_k \subseteq \vectorl_k
	\end{equation}
	
	Now, we prove the following claim.
	
	\begin{claim}
		\label{claim:linr}
		For any integer $k$ such that $1 \leq k \leq n$, if a point $p \in \vectorl_k$, then $p \in \vectorr_k$.
	\end{claim}	
	
	\begin{proof}
		%Let $\vectorw = B\vectorv$. 
		By the definition of $p$, there exists a vector of $d$-dimensional points $p'$ such that
		
		\begin{itemize}
			\item $\bfA_k  p' = p$, i.e., $p = \sum_{j = 1}^n \bfA_{kj} p'_{j}$;
			
			\item For $1 \leq j \leq n$, there exists a vector of $d$-dimensional points $p^j$ such that (i) $p'_j =  \bfB_j p^j$; and (ii) $p^j_i \in \vectorv_i$ for each $1 \leq i \leq n$. Note that condition (i) implies that $p'_j = \sum_{i = 1}^n \bfB_{ji} p^j_i$.
		\end{itemize}
		
		To prove the claim, we need to find a vector of $d$-dimensional points $q$ such that $(\bfA \bfB)_k q = p$ and $q_i \in \vectorv_i$ for each $1 \leq i \leq n$.  Define $q_i$ as follows:
		
		\begin{equation}
		\label{eq:matrix-q}
		q_i = \frac{\sum_{j = 1}^n \left( \bfA_{kj} \bfB_{ji} \right)p^j_i}{\sum_{j = 1}^n \bfA_{kj} \bfB_{ji}}
		\end{equation}
		
		Since by assumption, each $p^j_i \in \vectorv_i$, $q_i \in \vectorv_i$ as well. Now, we show that $(\bfA \bfB)_k q = p$.
		
		\begin{align*}
		(\bfA \bfB)_k q &= \sum_{i = 1}^n \left( \sum_{j=1}^n \bfA_{kj} \bfB_{ji} \right) q_i \\
		&= \sum_{i = 1}^n \left( \sum_{j=1}^n \bfA_{kj} \bfB_{ji} \right) \frac{\sum_{j = 1}^n \left( \bfA_{kj} \bfB_{ji} \right)p^j_i}{\sum_{j = 1}^n \bfA_{kj} \bfB_{ji}}\\
		&= \sum_{i = 1}^n \left( \sum_{j=1}^n \bfA_{kj} \bfB_{ji} \right) p^j_i\\
		&= \sum_{j = 1}^n \bfA_{kj} \left( \sum_{i = 1}^n \bfB_{ji} ~p^j_i \right)\\
		&= \sum_{j = 1}^n \bfA_{kj} p'_j\\
		&= p
		\end{align*}
		
		Hence, $p \in \vectorr_k$.
	\end{proof}
	
	Claim \ref{claim:linr} implies that
	
	\begin{equation}
	\label{eq:linr}
	\vectorl_k \subseteq \vectorr_k
	\end{equation}
	
	Equations (\ref{eq:rinl}) and (\ref{eq:linr}) together imply that for each $k$, 
	
	\[
	\vectorl_k = \vectorr_k
	\]
	
	Therefore, $\vectorl = \vectorr$. This completes the proof.
\end{proof}

\section{Proof of Theorem \ref{t_M}}
\label{a_t_M}

\noindent
{\bf Theorem \ref{t_M}:}
{\em 
For $t\geq 1$,
define $\vectorv[t]=\bfM[t]\vectorv[t-1]$, with $\bfM[t]$ as specified above.
Then, for $\tau\geq 0$, and for all $i\in V-\fbar[\tau+1]$,
% $h_i[t]$ is non-empty, and
$\vectorv_i[\tau]$ equals $h_i[\tau]$.
}

\begin{proof}
The proof of the above theorem is by induction on $\tau$.
Recall that
we defined $\vectorv_i[0]$ to be equal to $h_i[0]$ for
all $i\in V-\fbar[1]$ in the initialization step (I1) in Section \ref{s_alg}.
Thus, the theorem trivially holds for $\tau=0$.

Now, for some $\tau\geq 0$, and for
all $i\in V-\fbar[\tau+1]$, suppose that $\vectorv_i[\tau]=h_i[\tau]$.
% Thus, for all
% $i\in V-\fbar[t]$, $h_i[t-1]=\vectorv_i[t-1]\neq \emptyset$.
Recall that processes in $V-\fbar[\tau+2]$ surely survive at least till the end of round
$\tau+1$ (by definition of $\fbar[\tau+2]$).
Therefore, in round $\tau+1\geq 1$, each process in $i\in V-\fbar[\tau+2]$ computes
its new state 
$h_i[\tau+1]$
at line 14 of Algorithm CC, using function $\LL(~Y_i[\tau+1]~;~ [\frac{1}{|Y_i[\tau+1]|}, \cdots,\frac{1}{|Y_i[\tau+1]|}])$, where $Y_i[\tau+1]:=\{h~|~ (h,j,\tau+1)\in \msg_i[\tau+1]\}$. 
Also,
if $(h,j,\tau+1) \in \msg_i[\tau+1]$, then process $j$ 
must have sent round $\tau+1$ message $(h_j[\tau],j,\tau+1)$ to process $i$ -- in other words,
$h$ above (in $(h,j,\tau+1)\in \msg_i[\tau+1]$) must be equal to $h_j[\tau]$.
Also, since $j$ did send a round $\tau+1$ message, $j\in V-\fbar[\tau+1]$.
Thus, by induction hypothesis, $\vectorv_j[\tau]=h_j[\tau]$.

Now observe that, by definition of $Y_i[\tau+1]$ at line 13 of the algorithm,
$|Y_i[\tau+1]|=|\msg_i[\tau+1]|$. Thus, the definition of the matrix elements in (\ref{eq:matrix_i})
and (\ref{eq:matrix_i-2}) ensures that
$\bfM_i[\tau+1]\vectorv[\tau]$ equals $\LL(~Y_i[\tau+1]~;~ [\frac{1}{|Y_i[\tau+1]|}, \cdots,\frac{1}{|Y_i[\tau+1]|}])$, i.e., $h_i[\tau+1]$.
Thus, $\vectorv_i[\tau+1]$ defined as $\bfM_i[\tau+1]\vectorv[\tau]$ also equals $h_i[\tau+1]$.
This holds for all $i\in V-\fbar[\tau+2]$, completing the induction.
\end{proof}

\section{Useful Lemmas}
\label{a_lemmas}

In this section, we prove four lemmas used later in Appendix \ref{a_correctness}.

%\section{Proof of Theorem \ref{theorem:transition_matrix}}
%\label{a_t_matrix}

~

The procedure for constructing $\bfM[t]$ that the lemma below refers to is presented in Section \ref{ss_correctness}.
\begin{lemma}
\label{lemma:transition_matrix}
For $t\geq 1$,
transition matrix $\bfM[t]$ constructed using the above procedure satisfies the following conditions: 

\begin{itemize}
\item  $\bfM[t]$ is a row stochastic matrix.

\item  For $i,j \in V - \fbar[t+1]$, there exists a fault-free process $g(i,j)$ such that
 $\bfM_{ig(i,j)}[t] \geq \frac{1}{n}$ 
and 
 $\bfM_{jg(i,j)}[t] \geq \frac{1}{n}$ 
\end{itemize}

\end{lemma}

\begin{proof}
\begin{itemize}
\item Observe that, by construction, for each $i\in V$, the row vector $\bfM_i[t]$ 
contains only non-negative elements, which add up to 1. Thus, each row
$\bfM_i[t]$ is a stochastic vector, and hence
the matrix $\bfM[t]$ is row stochastic.

\item To prove the second claim in the lemma, consider any pair of processes $i,j \in V - \fbar[t+1]$.
Recall that the set $\msg_i[t]$ used in the construction of $\bfM[t]$ is such that $|\msg_i[t]|=|Y_i[t]|$ (i.e., $\msg_i[t]$ is
the message set at the point where $Y_i[t]$ is created).
Thus, $|\msg_i[t]|\geq n-f$ and $|\msg_j[t]| \geq n-f$, and there must be at least $n-2f$ messages in $\msg_i[t]\cap \msg_j[t]$. By assumption, $n \geq (d+2)f+1$. Hence, $n-2f \geq df+1 \geq f+1$, since $d \geq 1$. Therefore, there exists a fault-free process $g(i,j)$ such that $(h_{g(i,j)}[t-1],g(i,j),t) \in \msg_i[t] \cap \msg_j[t]$. By (\ref{eq:matrix_i}) in the procedure to construct $\bfM[t]$, $\matrixm_{ig(i,j)}[t]=\frac{1}{|\msg_i[t]|}\geq \frac{1}{n}$
and 
$\matrixm_{jg(i,j)}[t]=\frac{1}{|\msg_j[t]|}\geq\frac{1}{n}$.
\end{itemize}
\end{proof}

%++++++++ Moved from appendix H: Proof of Lemma 6 +++++++++

~

To facilitate the proof of next lemma below, we first
introduce some terminology and results related to matrices.

For a row stochastic matrix $\bfA$,
 coefficients of ergodicity $\delta(\bfA)$ and $\lambda(\bfA)$ are defined as
follows \cite{Wolfowitz}:
\begin{eqnarray*}
\delta(\bfA) & = &   \max_j ~ \max_{i_1,i_2}~ \| \bfA_{i_1\,j}-\bfA_{i_2\,j} \| \label{e_zelta} \\~\\
\lambda(\bfA) & = & 1 - \min_{i_1,i_2} \sum_j \min(\bfA_{i_1\,j} ~, \bfA_{i_2\,j}) \label{e_lambda}
\end{eqnarray*}
% It is easy to show that  $0\leq \delta(\bfA) \leq 1$ and $0\leq \lambda(\bfA) \leq 1$.
% We now state a useful result from \cite{Hajnal58}.
\begin{claim}
\label{claim_zelta}
For any $p$ square row stochastic matrices $\bfA(1),\bfA(2),\dots, \bfA(p)$, 
\begin{eqnarray*}
\delta(\Pi_{\tau=1}^p \bfA(\tau)) ~\leq ~
 \Pi_{\tau=1}^p ~ \lambda(\bfA(\tau)).
\end{eqnarray*}
\end{claim}
Claim \ref{claim_zelta} is proved in \cite{Hajnal58}.

~

\begin{claim}
\label{c_lambda_bound}
If there exists a constant $\gamma$, where $0<\gamma\leq 1$, such
that, for any pair of rows $i,j$ of matrix $\bfA$, there exists a column
$g$ (that may depend on $i,j$) such that,
$\min (\bfA_{ig},\bfA_{jg}) \geq \gamma$,
then
 $\lambda(\bfA)\leq 1-\gamma<1$.
\end{claim}
Claim \ref{c_lambda_bound} follows directly from the definition of $\lambda(\cdotp)$.

~

\begin{lemma}
\label{lemma:transition_matrix2}
For $t\geq 1$,
let $\bfP[t] = \Pi_{\tau=1}^t\, \bfM[\tau]$.
Then,
\begin{itemize}
\item $\bfP[t]$ is a row stochastic matrix.
\item For $i,j\in V-F$, and $k\in V$,
\begin{equation}
\|\, \bfP_{ik}[t] -  \bfP_{jk}[t]\,\| \leq \left(1-\frac{1}{n}\right)^t
\end{equation}
where $\|a\|$ denotes absolute value of real number $a$.
\end{itemize}
\end{lemma}

\begin{proof}
By the first claim of Lemma \ref{lemma:transition_matrix}, $\bfM[\tau]$ for $1\leq \tau \leq t$ is row stochastic. Thus, $\bfP[t]$ is a product of
row stochastic matrices, and hence, it is itself also
 row stochastic. 

Now, observe that by the second claim in Lemma \ref{lemma:transition_matrix} and
Claim \ref{c_lambda_bound}, $\lambda(\matrixm[t])\leq 1-\frac{1}{n}<1$.
Then by Claim \ref{claim_zelta} above,
\begin{eqnarray}
\label{e_alpha}
\delta(\bfP[t])=
\delta(\Pi_{\tau=1}^t \bfM[\tau])
~ \leq ~  \Pi_{\tau=1}^{t} \lambda(\bfM[\tau]) 
 ~\leq~ {\left(1-\frac{1}{n}\right)}^t %~=~\alpha^t
\end{eqnarray}

Consider any two fault-free processes $i,j\in V-F$.
By (\ref{e_alpha}), $\delta(\bfP[t])\leq {\left(1-\frac{1}{n}\right)}^t$.
Therefore, by the definition of $\delta(\cdot)$, for $1 \leq k \leq n$, we have

\begin{equation}
\| \bfP_{ik}[t] -  \bfP_{jk}[t]\| \leq {\left(1-\frac{1}{n}\right)}^t
\end{equation}
\end{proof}

~

We now prove two lemmas related to validity of convex hulls computed in Algorithm CC. Recall that a valid convex hull is defined in Definition \ref{def:valid_hull}.

%\section{Proof of Lemma \ref{lemma:valid_initial_hull}}
%\label{app_s_lemma:valid_initial_hull}

%\noindent{\bf Lemma \ref{lemma:valid_initial_hull}:} {\em $h_i[0]$ for each process $i \in V - \fbar[1]$ is valid. \\ }

\begin{lemma}
\label{lemma:valid_initial_hull}
$h_i[0]$ for each process $i \in V - \fbar[1]$ is valid.
\end{lemma}

\begin{proof}
Recall that $h_i[0]$ is obtained on line 5 of Algorithm CC as \[ h_i[0]~:=~\cap_{\,C \subseteq X_i,\, |C| = |X_i| - f}~~\HH(C),\] where $X_i = \{\,x \,|\,(x,k,0)\in R_i\}$.
Under the {\em crash faults with incorrect inputs} model,
except for up to $f$ values in $X_i$ (which may correspond to
inputs at faulty processes), all the
other values in $X_i$ must correspond to inputs at fault-free processes (and hence
they are correct).
Therefore, at least one set $C$ used in the computation of $h_i[0]$ 
must contain only the inputs at fault-free processes. 
Therefore, $h_i[0]$ is in the convex hull of the inputs at fault-free processes.
That is, $h_i[0]$ is valid.
\end{proof}

\begin{lemma}
\label{lemma:linear_valid}
Suppose non-empty convex polytopes $h_1, h_2, \cdots, h_\nu$ are all valid. Consider $\nu$ constants $c_1, c_2, \cdots, c_\nu$ such that $0 \leq c_i \leq 1$ and $\sum_{i = 1}^\nu c_i = 1$.
Then the linear combination of these convex polytopes,
$\LL([h_1, h_2, \cdots, h_\nu]\,;\, [ c_1, c_2, \cdots, c_\nu])$, is
{\bf convex, non-empty, and valid}.
\end{lemma}

\begin{proof}
Polytopes $h_1,\cdots,h_\nu$ are given as non-empty, convex, and valid.
Let
\begin{equation}
\label{e_a_l}
L := \LL([h_1, h_2,\cdots, h_\nu]\,;~[c_1, c_2,\cdots, c_\nu])
\end{equation}
We will show that $L$ is convex, non-empty, and valid.

\paragraph{\em $L$ is convex:}
Given any two points $x, y$ in $L$, by Definition \ref{def:linear_hull}, we have
\begin{equation}
\label{eq:x2}
x = \sum_{1\leq i\leq \nu} c_i p_{(i, x)}~~\text{for some}~p_{(i,x)} \in h_i, 1\leq i\leq \nu
\end{equation}

and 

\begin{equation}
\label{eq:y2}
y = \sum_{1\leq i\leq \nu} c_i p_{(i,y)}~~\text{for some}~p_{(i,y)} \in h_i,~~1\leq i\leq \nu
\end{equation}

Now, we show that any convex combination of $x$ and $y$ is also in $L$ defined in (\ref{e_a_l}). Consider a point $z$ such that

\begin{equation}
\label{eq:z2}
z = \theta x + (1 - \theta) y~~~\text{where}~0 \leq \theta \leq 1
\end{equation}

Substituting (\ref{eq:x2}) and (\ref{eq:y2}) into (\ref{eq:z2}), we have
\begin{align}
z &= \theta \sum_{1\leq i\leq\nu}~c_i p_{(i,x)} + (1-\theta) \sum_{1\leq i\leq \nu}~c_i p_{(i,y)}\nonumber\\
&= \sum_{1\leq i\leq \nu}~c_i\left(\theta p_{(i,x)} + (1-\theta) p_{(i,y)}\right) \label{eq:p_iz2}
\end{align}

Define $p_{(i,z)} = \theta p_{(i,x)} + (1-\theta) p_{(i,y)}$ for $1\leq i\leq \nu$. Since $h_i$ is convex, and $p_{(i,z)}$ is a convex combination of $p_{(i,x)}$ and $p_{(i,y)}$, $p_{(i,z)}$ is also in $h_i$. Substituting the definition of $p_{(i,z)}$ in (\ref{eq:p_iz2}), we have

\begin{align*}
z &=  \sum_{1\leq i\leq \nu}~~c_i~p_{(i,z)}~~\text{where}~p_{(i,z)} \in h_i,~~1\leq i\leq \nu
\end{align*}
Hence, by Definition \ref{def:linear_hull}, $z$ is also in $L$. Therefore, $L$ is convex.

\paragraph{\em $L$ is non-empty:}
The proof that $L$ is non-empty is trivial. Since each of the $h_i$'s is non-empty,
there exists at least one point $z_i\in h_i$ for $1\leq  i \leq \nu$.
Then $\sum_{1\leq i\leq \nu} c_iz_i$ is in $L$, and hence $L$ is non-empty.

\paragraph{\em $L$ is valid:}
The proof that $L$ is valid is also straightforward. Since each of the $h_i$'s
is valid, each point in each $h_i$ is a convex combination of the inputs
at the fault-free processes. Since each point in $L$ is a convex combination
of points in $h_i$'s, it then follows that each point in $L$
is in the convex hull of the inputs at fault-free processes.
\end{proof}

\comment{++++++++++++++++++++++++++++
\begin{proof}
Polytopes $h_1,\cdots,h_\nu;c_1,\cdots,c_\nu$ are given as non-empty and convex.
Observe that the points in $\LL(h_1,\cdots,h_\nu;c_1,\cdots,c_\nu)$ are convex
combinations of the points in $h_1,\cdots,h_\nu$, because $\sum_{i=1}^\nu c_i = 1$ and $0\leq c_i\leq 1$,
for $1\leq i\leq \nu$.
Let $I$ be the set of inputs at the fault-free processes
in $V-F$. Then, $\HH(I)$ is the convex hull of the inputs at the fault-free processes.
Since $h_i$, $1\leq i\leq k$, is valid, each point $p\in h_i$ is in $\HH(I)$.
Since $\HH(I)$ is a convex polytope, it follows that any convex combination
of the points in $h_1,\cdots,h_k$ is also in $\HH(I)$.
\end{proof}
++++++++++++++++++++++}

\section{Proof of Theorem \ref{thm:correctness}}
\label{a_correctness}

~

\noindent
{\bf Theorem \ref{thm:correctness}:}
{\em 
% Algorithm CC is correct.
Algorithm CC satisfies the {\em validity}, {\em $\epsilon$-agreement} and {\em termination} properties.
}

~

\begin{proof}
We prove that Algorithm CC satisfies the {\em validity}, {\em $\epsilon$-agreement} and {\em termination} properties after a large enough number of asynchronous rounds.

Repeated applications of
Lemma \ref{l_progress} ensures that the fault-free processes will progress
from round 0 through round $r$, for any $r\geq 0$, allowing us to use (\ref{e_unroll}).
Consider round $t \geq 1$. Let
\begin{eqnarray}\bfP[t] &=& \Pi_{\tau=1}^t \bfM[\tau]
\label{e_Mstar}
\end{eqnarray}
% (To simplify the presentation, we do not include the round index $[t]$
% in the notation $\bfM^*$ above.)
% Then $\vectorv[t]=\bfM^* \vectorv[0]$.

\paragraph{\normalfont\em Validity:}
We prove validity using the series of observations below:
\begin{itemize}
\item {\em Observation 1}:
By Lemma \ref{lemma:J_in_H0} (in Appendix \ref{a_lemma_J}) and Lemma \ref{lemma:valid_initial_hull} (in Appendix \ref{a_lemmas}),
$h_i[0]$ for each $i\in \sv - \fbar[1]$ is non-empty and valid.
Also, each such $h_i[0]$ is convex by construction (line 5 of Algorithm CC).

\item {\em Observation 2}:
As per the initialization step (I1) (in Section \ref{ss_correctness}), for each $i\in V-\fbar[1]$,
$\vectorv_i[0] := h_i[0]$; thus, by Observation 1 above, for each such process $i$,
$\vectorv_i[0]$ is convex, valid and non-empty.
Also, in initialization step (I2) (in Section \ref{ss_correctness}), for each process $k\in \fbar[1]$,
we set $\vectorv_k[0]:=h_m[0]$, where $m$ is a fault-free process;
thus, by Observation 1, for each such process $k$, $\vectorv_k[0]$ is convex, valid and non-empty.
Therefore, each element of $\vectorv[0]$ is a 
non-empty, convex and valid polytope.

\item {\em Observation 3}: By Lemma \ref{lemma:transition_matrix2} in Appendix \ref{a_lemmas}, $\bfP[t]$ is a {\em row stochastic} matrix.
Thus, elements of each row of $\bfP[t]$ are non-negative and add up to 1.
Therefore,
by Observation 2 above, and Lemma \ref{lemma:linear_valid} in Appendix \ref{a_lemmas}, $\bfP_i[t]\vectorv[0]$ for each $i\in V-F$
is valid, convex and non-empty. Also, by Theorem \ref{t_M}, and equation (\ref{e_unroll}), $h_i[t]=\bfP[t]\vectorv[0]$
for $i\in V-F$.
Thus, $h_i[t]$ is valid, convex and non-empty for $t \geq 1$.
\end{itemize}
Therefore, {\em Algorithm CC} satisfies the validity property. 

\paragraph{\normalfont \em $\epsilon$-Agreement and Termination:}

Recall that by Lemma \ref{lemma:transition_matrix2} in Appendix \ref{a_lemmas}, for any two fault-free processes $i,j\in V-F$, and for $1 \leq k \leq n$, we have

\[
\| \bfP_{ik}[t] -  \bfP_{jk}[t]\| \leq \left(1-\frac{1}{n}\right)^t
\]

Processes in $\fbar[1]$ do not send any messages to
any other process in round 1 and beyond. Thus, by the construction of $\bfM[t]$,
for each $a\in V-\fbar[1]$  and $b\in \fbar[1]$,
$\bfM_{ab}[t]=0$ for all $t\geq 1$;
it then follows that $\bfP_{ab}[t]=0$ as well.\footnote{Claim \ref{claim:M*}
in Appendix \ref{a_l:svSize} below provides a more detailed proof of this statement.}

Consider fault-free processes $i,j\in V-F$.
(In the following discussion,
we will denote a point in the $d$-dimensional Euclidean space by a list of its $d$ coordinates.)
The previous paragraph implies that, for
any point $q_i$ in $h_i[t]=\vectorv_i[t]=\bfP_i[t]\vectorv[0]$,
there must exist, for all $k\in V-\fbar[1]$, $p_k \in h_k[0],$ such that
\begin{equation}
\label{pi}
q_i = \sum_{k\in V-\fbar[1]} \bfP_{ik}[t] p_k
= \left(\sum_{k\in V-\fbar[1]} \bfP_{ik}[t] p_k(1),~~\sum_{k\in V-\fbar[1]} \bfP_{ik}[t] p_k(2), \cdots, \sum_{k\in V-\fbar[1]} \bfP_{ik}[t] p_k(d)\right)
\end{equation}
where $p_k(l)$ denotes the value of $p_k$'s $l$-th coordinate.
The list on the right-hand-side of the above equation represents
the $d$ coordinates of point $p_i$.

Using points $p_k$ in the above equation,
now choose point $q_j$ in $h_j[t]$ defined as follows.
\begin{equation}
\label{pj}
q_j = \sum_{k\in V-\fbar[1]} \bfP_{jk}[t] p_k 
= \left(\sum_{k\in V-\fbar[1]} \bfP_{jk}[t] p_k(1),~~ \sum_{k\in V-\fbar[1]} \bfP_{jk}[t] p_k(2), \cdots,  \sum_{k\in V-\fbar[1]} \bfP_{jk}[t] p_k(d)\right)
\end{equation}

Recall that the Euclidean distance between $q_i$ and $q_j$ is $\ddd(q_i,q_j)$. From Lemma \ref{lemma:transition_matrix2} (in Appendix \ref{a_lemmas}), (\ref{pi}) and (\ref{pj}), we have
the following:
%The following derivation is obtained by simple algebraic manipulation, using Lemma \ref{lemma:transition_matrix2}, (\ref{pi}) and (\ref{pj}).

\begin{eqnarray*}
\ddd(q_i, q_j) &=&  \sqrt{\sum_{l=1}^d (q_i(l) - q_j(l))^2} \nonumber\\
&= &\sqrt{\sum_{l=1}^d \left(\sum_{k\in V-\fbar[1]} \bfP_{ik}[t] p_k(l) - \sum_{k\in V-\fbar[1]} \bfP_{jk} p_k(l)\right)^2}\nonumber ~~~~ \text{by (\ref{pi}) and (\ref{pj}})\\
&= &\sqrt{\sum_{l=1}^d \left(\sum_{k\in V-\fbar[1]} (\bfP_{ik}[t]-\bfP_{jk}[t]) p_k(l)\right)^2}\nonumber\\
&\leq &\sqrt{\sum_{l=1}^d \left[\left(1-\frac{1}{n}\right)^{2t}\left( \sum_{k\in V-\fbar[1]} \|p_k(l)\|\right)^2\right]} 
~~~~~~~ \mbox{~~~~ by Lemma \ref{lemma:transition_matrix2}}\nonumber\\
&= &\left(1-\frac{1}{n}\right)^t \sqrt{\sum_{l=1}^d \left(\sum_{k\in V-\fbar[1]} \|p_k(l)\|\right)^2}\label{eq:d_pipj}
\end{eqnarray*} 

Define 
\[
\Omega = \max_{p_k \in h_k[0],k\in V-\fbar[1]} \sqrt{\sum_{l=1}^d (\sum_{k\in V-\fbar[1]} \|p_k(l)\|)^2}
\]
Then, we have 
\begin{equation}
\label{e_o}
\ddd(q_i, q_j) \leq \left(1-\frac{1}{n}\right)^t \sqrt{\sum_{l=1}^d \left(\sum_{k\in V-\fbar[1]} \|p_k(l)\|\right)^2} ~\leq~ (1-\frac{1}{n})^t\,\Omega
\end{equation}

Because the $h_k[0]$'s in the definition of $\Omega$ are all valid (by
Lemma \ref{lemma:valid_initial_hull} in Appendix \ref{a_lemmas}),
$\Omega$ can itself be upper bounded by a function of the input vectors at the fault-free
processes.
In particular, under the assumption that each element of fault-free processes' input vectors is
upper bounded by $U$ and lower bounded by $\mu$, $\Omega$ is upper
bounded by $\sqrt{dn^2\max(U^2,\mu^2)}$.
Observe that the upper bound on the right-hand-side of (\ref{e_o})  monotonically decreases with $t$, because $1-\frac{1}{n}<1$.
Define $t_{end}$ as the smallest positive integer $t$
for which
\begin{equation}\left(1-\frac{1}{n}\right)^t\sqrt{dn^2\max(U^2,\mu^2)}<\epsilon \label{e_end2} \end{equation}
Recall that the algorithm terminates after $t_{end}$ rounds.
Since $t_{end}$ is finite, the algorithms satisfies the {\em termination}
condition.

(\ref{e_o}) and (\ref{e_end2}) together imply that, 
for fault-free processes $i,j$ and for each point $q_i\in h_i[t_{end}]$,
there exists a point $q_j[t]\in h_j[t_{end}]$, such that $\ddd(q_i,q_j)<\epsilon$ (and, similarly, vice-versa).
Thus, by Definition of Hausdorff distance, $\D(h_i[t_{end}],h_j[t_{end}]) <\epsilon$.
Since this holds true for any pair of fault-free processes
$i,j$, the $\epsilon$-agreement property is satisfied at termination.
\end{proof}

\comment{+++++++++++ replaced by Proof of Theorem 3 ++++++++
\section{Proof of Lemma \ref{lemma:alpha_t}}
\label{a_alpha_t}

++++++++ Definition of Matrix tools moved +++++++++

Now, we are ready to prove Lemma \ref{lemma:alpha_t}:

\noindent{\bf Lemma \ref{lemma:alpha_t}:} {\em
For any two fault-free processes $i,j\in V-F$, and for $1 \leq k \leq n$,
\begin{equation}
\| \bfM^*_{ik} -  \bfM^*_{jk}\| \leq \alpha^t
\end{equation}
}

\begin{proof}
First, observe that by the first claim in Lemma \ref{lemma:transition_matrix} and
Claim \ref{c_lambda_bound}, $\lambda(\matrixm[t])\leq 1-\frac{1}{n}<1$.

Then by Claim \ref{claim_zelta},
\begin{eqnarray}
\label{e_alpha}
\delta(\bfM^*)=
\delta(\Pi_{\tau=1}^t \bfM[\tau])
~ \leq ~ \lim_{t\rightarrow\infty} \Pi_{\tau=1}^{t} \lambda(\bfM[\tau]) 
 ~\leq~ {\left(1-\frac{1}{n}\right)}^t ~=~\alpha^t
\end{eqnarray}

Consider any two fault-free processes $i,j\in V-F$.
By (\ref{e_alpha}), $\delta(\bfM^*)\leq \alpha^t$. 
Therefore, by the definition of $\delta(\cdot)$, for $1 \leq k \leq n$, we have

\begin{equation}
\| \bfM^*_{ik} -  \bfM^*_{jk}\| \leq \alpha^t
\end{equation}
\end{proof}
+++++++++++++++++++++=}

\section{Proof of Lemma \ref{lemma:svSize}}
\label{a_l:svSize}

We first prove a claim that will be used in the proof of Lemma \ref{lemma:svSize}.
\begin{claim}
\label{claim:M*}
For $t \geq 1$, let $\bfP[t] = \Pi_{\tau=1}^t \bfM[\tau]$. Then, for all processes $j \in V - \fbar[t+1]$, and $k \in \fbar[1]$, $\bfP_{jk}[t] = 0$.
\end{claim}

\begin{proof}
The claim is intuitively straightforward.
For completeness, we present a formal proof here.
The proof is by induction on $t$.

{\em Induction Basis}: Consider the case when $t = 1$, $j \in V - \fbar[2]$, and $k \in \fbar[1]$. Then by definition of $\fbar[1]$, $(*,k,0) \not\in \msg_j[1]$.
Then, due to (\ref{eq:matrix_i-2}),
$\bfM_{jk}[1]=0$, and hence $\bfP _{jk}[1]=\bfM_{jk}[1]=0$.

{\em Induction}: Consider $t \geq 2$. Assume that the claim holds true through round  $t - 1$. 
Then, $\bfP _{jk}[t-1]=0$ for all $j \in V - \fbar[t]$ and $k \in \fbar[1]$.
Recall that $\bfP [t-1] = \Pi_{\tau=1}^{t-1} \bfM[\tau]$.

Now, we will prove that the claim holds true for round $t$.
Consider $j \in V - \fbar[t+1]$
and $k\in \fbar[1]$.
Note that $\bfP [t] = \Pi_{\tau=1}^{t} \bfM[\tau] = \bfM[t] \Pi_{\tau=1}^{t-1} \bfM[\tau] = \bfM[t] \bfP [t-1]$.
Thus, $\bfP _{jk}[t]$ can be non-zero only if there exists a $q \in V$ such that $\bfM_{jq}[t]$ and $\bfP _{qk}[t-1]$
are both non-zero.

For any $q\in \fbar[t-1]$, $(*,q,t-1) \not\in \msg_j[t]$.
Then, due to (\ref{eq:matrix_i-2}), $\bfM_{jq}[t] = 0$ for all $q \in \fbar[t-1]$, and hence all $q\in \fbar[1]$ (note that $\fbar[r-1]\subseteq \fbar[r]$
for $r\geq 2$).
Additionally, by the induction hypothesis,
for all $q \in V - \fbar[t]$ and $k \in \fbar[1]$,
 $\bfP _{qk}[t-1] = 0$.
Thus, these two observations together imply that there does not exist any $q \in V$ such that
$\bfM_{jq}[t]$ and $\bfP _{qk}[t-1]$ are both non-zero.
Hence, $\bfP _{jk}[t]=0$.
\end{proof}

~

~

\noindent
Now, we are ready to prove Lemma \ref{lemma:svSize}.

~

\noindent{\bf Lemma \ref{lemma:svSize}:}
{\em
For all $i\in V-\fbar[t+1]$ and $t \geq 0$,
 $I_Z \subseteq h_i[t]$.
}

\begin{proof}
Recall that $Z$ and $I_Z$ are defined in (\ref{e_Z}) and (\ref{e_I_Z}), respectively.
We first prove that for all $j\in V-\fbar[1]$,
 $I_Z \subseteq h_j[0]$. We make the following observations for each process $i \in V-\fbar[1]$:

\begin{itemize}
% \item {\em Observation 1}: By the {\em liveness}
% and {\em containment} properties of {\em stable vector},
% $Z$ contains at least $n-f$ messages.
% 
\item {\em Observation 1}: By the definition of multiset $X_i$ at line 4
of round 0 at process $i$, and the definition of $X_Z$ in Section \ref{s_optimal}, we have $X_Z\subseteq X_i$.

\item {\em Observation 2}: Let $A$ and $B$ be sets of points in the $d$-dimensional space, where $|A|\geq n-f$, $|B| \geq n-f$ and $A \subseteq B$. Define $h_A := \cap_{\,C_A \subseteq A, |C_A| = |A| - f}~~\HH(C_A)$ and $h_B := \cap_{\,C_B \subseteq B, |C_B| = |B| - f}~~\HH(C_B)$. Then $h_A \subseteq h_B$. This observation follows directly from the fact that every multiset $C_A$ in the computation of $h_A$ is contained in some multiset $C_B$ used in the computation of $h_B$, and the property of $\HH$.
\end{itemize}

Now, consider the computation of $h_i[0]$ at line 5. By Observations 1 and 2,
and the definitions of $h_i[0]$ and $I_Z$, we have that  $I_Z \subseteq h_i[0]=\vectorv_i[0]$, where $i\in V-\fbar[1]$.
Also, by initialization step (I2) (in Section \ref{ss_correctness}), for $k\in \fbar[1]$, $\vectorv_k[0]=h_m[0]$, for
some fault-free process $m$. Thus, all the elements of $\vectorv[0]$
contain $I_Z$.
Then, due to row stochasticity of $\Pi_{\tau=1}^t\bfM[\tau]$,
it follows that each element of $\vectorv[t]=\left(\Pi_{\tau=1}^t\bfM[\tau\right)\,\vectorv[0]$ also contain $I_Z$.
Recalling that $h_i[t]=\vectorv_i[t]$ for each fault-free process,
proves the claim of the lemma.

\comment{+++++++++++++++++++++
\noindent
Now we make several observations for each fault-free process $i\in V-F$:
\begin{itemize}
\item As shown above, $I_Z\subseteq h_j[0]$ for all $j\in V-\fbar[1]$. 
\item
From (\ref{e_Mstar}), for $t\geq 1$,
\[
\vectorv[t]=\bfP[t] \vectorv[0]
\]
where $\vectorv_j[0]=h_j[0]$ for $j\in V-\fbar[1]$.
\item By Theorem \ref{t_M},  $\vectorv_i[t] = h_i[t]$.
\item Due to Claim \ref{claim:M*}, $\bfP_{ik}[t]=0$ for $k\in\fbar[1]$ (i.e., $k\not\in V-\fbar[1]$).
\item

$\bfM^*$ is the product of row stochastic matrices; therefore, $\bfM^*$ itself is also row stochastic.
Thus, for fault-free process $i$, $\vectorv_i[t]=h_i[t]$ is obtained
as the product of the $i$-th row of $\bfM^*$, namely $\bfM^*_i$, and $\vectorv[0]$:
this product yields 
a linear combination of the elements of $\vectorv[0]$, where the weights
are non-negative and add to 1 (because $\bfM^*_i$ is a stochastic row vector).

\item
From (\ref{e_r_c}), recall that $\bfM_i^*\vectorv[0]=\LL(\vectorv[0]^T~;~\bfM^*_i)$.
Function $\LL$ ignores the input polytopes for which the corresponding weight is 0.
Finally, from the previous observations, we have that when the weight in $\bfM^*[i]$ 
is non-zero, the corresponding polytope in 
$\vectorv[0]^T$ contains $I_Z$. 
Therefore, the linear combination also contains $I_Z$.
\end{itemize}
Thus,
$I_Z$ is contained in $h_i[t]=\vectorv_i[t]=\bfM^*_i\vectorv[0]$.
+++++++++++++++++++++++++}
\end{proof}

\section{Proof of Theorem \ref{thm:optSize}}
\label{a_t:optSize}

\noindent{\bf Theorem \ref{thm:optSize}:}
{\em
{\em Algorithm CC} is optimal under the notion of optimality in Section \ref{s_intro}.
}

\begin{proof}
Consider multiset $X_Z$ defined in Section \ref{s_optimal}.
Recall that $|X_Z|=|Z|$, and that $Z$ contains at least 
$n-f$ tuples. Thus, $X_Z$ contains at least $n-f$ points, and of these at least $n-2f$ points must be the inputs at fault-free processes.
Let $V_Z$ denote the set of fault-free processes whose inputs appear in $X_Z$.
Let $S=V-F-V_Z$. Since $|X_Z|\geq n-f$, $|S|\leq f$. 

Now consider the following execution of any algorithm ALGO that correctly solves approximate convex consensus. 
Suppose that the faulty processes in $F$ do not crash, but have an incorrect input. Consider the case when processes in $S$ are so slow that the other fault-free processes must terminate before receiving any messages from the processes in $S$. The fault-free processes in $V_Z$ cannot determine whether the processes in $S$ are just slow, or they have crashed.

Processes in $V_Z$ must be able to terminate without receiving any messages from the processes in $S$. Thus, their output must be in the convex
hull of inputs at the fault-free processes whose inputs are included in $X_Z$. However, any $f$ of the processes whose inputs are
in $X_Z$ may potentially be faulty and have incorrect inputs. Therefore, the output obtained by ALGO must be contained
in $I_Z$ as defined in Section \ref{s_optimal}. 
On the other hand, by Lemma \ref{lemma:svSize} in Appendix \ref{a_l:svSize}, the output obtained using Algorithm CC contains $I_Z$.
This proves the theorem.
\end{proof}

\section{Convex Function Optimization}
\label{a_optimization}

\subsection{Notion of Optimality}

We can extend the notion of optimality (of {\em convex consensus} algorithms) in Section \ref{s_intro} to {\em convex function optimization}
as follows. An algorithm $A$ for convex function optimization is said to be optimal if the following condition is true.
\begin{list}{}{}
\item[]
Let $F$ denote a set of up to $f$ faulty processes. For a {\bf given execution} of algorithm $A$ with $F$ being the set of faulty processes, let $y_i(A)$ denote the output at process $i$ at the end of the given execution. 
For any other algorithm $B$,
{\bf there exists} an execution with $F$ being the set of faulty processes, such that $y_i(B)$ is the output at fault-free process $i$, and $c(y_j(A)) \leq c(y_j(B))$ for {\bf each} fault-free process $j$.
\end{list}
The intuition behind the above formulation is as follows. A goal of function optimization here is to
allow the processes to ``learn'' the smallest value of the cost function over the convex hull of the inputs
at the fault-free processes. The above condition implies that an optimal algorithm will learn a function value
that is no larger than that learned in a worst-case execution of any other algorithm.

The 2-step convex function optimization algorithm, with the first step being convex consensus, as described
in Section \ref{s_optimization}, is optimal in the above sense. This is a direct consequence of
Theorem \ref{thm:optSize}.

\subsection{Impossibility Result}

The four properties for convex function optimization problem introduced in Section \ref{s_optimization} are:

\begin{itemize}
\item {\bf Validity}: 
output $y_i$ at fault-free process $i$ is a
point in the convex hull of the correct inputs.

\item \textbf{$\epsilon$-Agreement}: for a given constant $\epsilon > 0$,
for any fault-free processes
$i,j$, $\ddd(y_i,y_j)<\epsilon$.

\item \textbf{Weak $\beta$-Optimality}:
(i)
for any constant $\beta>0$, for any fault-free processes $i,j$,
$\| c(y_i)-c(y_j)\|<\beta$, and
(ii) if at least $2f+1$ processes (faulty or fault-free) have an identical input, say $x$,
then for any fault-free process $i$, $c(y_i) \leq c(x)$.

\item {\bf Termination:} each fault-free process must terminate within a finite amount of time.
\end{itemize}

% In this appendix, we will prove
% prove that the above four conditions
% cannot be satisfied simultaneously
% in an asynchronous system in the presence of crash faults (with
% or without correct inputs).

The theorem below proves the impossibility of satisfying the above
properties for $n\geq 4f+1$ and $d\geq 1$. From our prior discussion,
we know that we need $n\geq (d+2)f+1$ even without the
weak $\beta$-optimality requirement. Thus, for $d\geq 2$,
the theorem implies that for $d\geq 2$ and any $n$, the above
properties cannot be satisfied. For the specific case of $d=1$,
we do not presently know whether the above properties can be
satisfied when $3f+1\leq n\leq 4f$.

\begin{theorem}
\label{thm:impossible}
All the four properties above cannot be satisfied simultaneously
in an asynchronous system in the presence of crash faults with
incorrect inputs for $n\geq 4f+1$ and $d\geq 1$.
\end{theorem}

\begin{proof}
We will prove the result for $d=1$. It should be obvious that impossibility
with $d=1$ implies impossibility for larger $d$ (since we can always choose
inputs that have 0 coordinates in all dimensions except one).

The proof is by contradiction. Suppose that there exists an algorithm,
say Algorithm $\sa$,
that achieves the above four properties for $n\geq 4f+1$ and $d=1$.

Let the cost function be given by $c(x) = 4-(2x-1)^2$ for $x\in [0,1]$
and $c(x)=3$ for $x\not\in[0,1]$. 
For future reference note that within the interval
$[0,1]$, function $c(x)$ has the smallest value at
$x=0,1$ both. 

Now suppose that all the inputs (correct and incorrect) are restricted to be binary, and must be 
0 or 1. (We will prove impossibility under this restriction on the inputs at faulty
and fault-free processes both,
which suffices to prove that the four properties cannot {\em always} be satisfied.)
Suppose that the output of Algorithm $\sa$ at fault-free process
$i$ is $y_i$. Due to the validity property, and because the inputs
are restricted to be 0 or 1, we know that $y_i\in [0,1]$.
% In fact, because $n\geq 4f+1$ and inputs (correct and incorrect)
% are restricted to be 0 or 1, part (ii) of the weak $\beta$-optimality
% condition will also apply.

Since $\lceil\frac{n}{2}\rceil\geq \lceil \frac{4f+1}{2} \rceil = 2f+1$,
at least $2f+1$ processes will have either input $0$, or input $1$. Without loss of generality,
suppose that at least $2f+1$ processes have input $0$.

Consider a fault-free process $i$. By weak $\beta$-Optimality, $c(y_i) \leq c(0)$,
that is, $c(y_i)\leq 3$. However, the minimum value of the cost function is 3 over
all possible inputs. Thus, $c(y_i)=3$.
Similarly, for any other fault-free process $j$ as well, $c(y_j)$ must equal 3. Now,
due to validity, $y_j\in [0,1]$, and the cost function is 3 in interval $[0,1]$ only
at $x=0,1$. Therefore, we must have $y_i$ equal to $0$ or $1$, and $y_j$ also equal to 0 or 1.
However, because algorithm $\sa$ satisfies the $\epsilon$-agreement condition,
$\ddd(y_i,y_j)=\| y_i-y_j\|<\epsilon$ (recall that dimension $d=1$). If
$\epsilon<1$, then $y_i$ and $y_j$ must be identical (because we already know
that they are either 0 or 1).
Since this condition holds for any pair of fault-free processes, it implies
{\em exact} consensus. Also, $y_i$ and $y_j$ will be equal to the input at a fault-free process
due to the validity property above, and because the inputs are restricted to be 0 or 1.
In other words, Algorithm $\sa$ can be used to solve exact consensus in the presence
of crash faults with incorrect inputs when $n\geq 4f+1$ in an asynchronous system.
This contradicts the well-known impossibility result by
Fischer, Lynch, and Paterson \cite{impossible_proof_lynch}.
\end{proof}

\end{document}